\pdfoutput=1
\documentclass[a4paper,onecolumn,nopdfoutputerror,accepted=2025-10-17]{quantumarticle}
\usepackage{CSTheoryToolkitCMUStyle}
\usepackage{Custom}
\usepackage{wasysym}
\usepackage{array}
\usepackage{bookmark}
\usepackage{quantikz}
\usepackage[nameinlink]{cleveref}

\makeatletter
\renewcommand{\@seccntformat}[1]{%
  \ifcsname prefix@#1\endcsname
    \csname prefix@#1\endcsname
  \else
    \csname the#1\endcsname\quad
  \fi}
\makeatother

\title{Effective Distance of Higher Dimensional HGPs and Weight-Reduced Quantum LDPC Codes}

\author[1]{Shi Jie Samuel Tan}
\author[1]{Lev Stambler}
\affiliation[1]{Joint Center for Quantum Information and Computer Science, University of Maryland, College Park, MD 20742, USA}
\affiliation[1]{Department of Computer Science, University of Maryland, MD 20742, USA}

\begin{document}
\newcommand{\code}{\mathcal{C}}
\newcommand{\parMat}{H}
\newcommand{\colWeight}{\Delta_v}
\newcommand{\rowWeight}{\Delta_c}
\newcommand{\factGraph}{\calG}

\newcommand{\bitIdx}{b}
\newcommand{\checkIdx}{s}
\newcommand{\bitDirection}{D}
\newcommand{\bitCopy}{\bitIdx_\bitDirection^i}
\newcommand{\stabGauge}{\checkIdx_a^{j}}

\newcommand{\gaugeBit}{\bar{\bitIdx}}
\newcommand{\bitsNotEnd}{B(\texttt{Body})}
\newcommand{\bitsEnd}{B(\texttt{End})}

\newcommand{\chainComplex}{\mathcal{C}}
\newcommand{\complex}{\mathcal{A}}
\newcommand{\Hx}{H_x}
\newcommand{\Hz}{H_z}
\newcommand{\DiamondComplex}{\blacklozenge}
\newcommand{\DiamondComplexGeneric}{\blacklozenge(\bitIdx, i, i + 1)}

\newcommand{\DiamondComplexEmpty}{\lozenge}
\newcommand{\Degenerator}{{\CIRCLE} \;}
\newcommand{\DegeneratorNS}{{\CIRCLE}}
\newcommand{\new}[1]{\widetilde{#1}}
\newcommand{\row}{\text{row}}

\newcommand{\updatedHx}{\widetilde{\Hx}}
\newcommand{\updatedHz}{\widetilde{\Hz}}
\newcommand{\complexNeighb}{\Gamma}

 \newcommand{\LEV}[1]{\textcolor{red}{[LEV: #1]}}
 \newcommand{\SAM}[1]{\textcolor{green}{[SAM: #1]}}

\algrenewcommand\algorithmicrequire{\textbf{Input:}}
\algrenewcommand\algorithmicensure{\textbf{Output:}}

\newcommand{\prunedT}{\widetilde{\calT}}
\newcommand{\col}{\text{col}}

\newcommand{\im}{\text{im}}

\newcommand{\padArrow}[1]{\xrightarrow{\hspace{0px}{#1}\hspace{0px}}}


\date{}
\maketitle

\begin{abstract}
  Quantum error correction plays a prominent role in the realization of quantum computation, and quantum low-density parity-check (qLDPC) codes are believed to be practically useful stabilizer codes.
  While qLDPC codes are defined to have constant weight parity-checks, the weight of these parity checks could be large constants that make implementing these codes challenging.
  Large constants can also result in long syndrome extraction times and bad error propagation that can impact error correction performance.
  Hastings recently introduced weight reduction techniques for qLDPC codes that reduce the weight of the parity checks as well as the maximum number of checks that acts on any data qubit~\cite{hastings2021quantum}.
  However, the fault tolerance of these techniques remains an open question.
  In this paper, we analyze the effective distance of the weight-reduced code when single-ancilla syndrome extraction circuits are considered for error correction.
  We prove that there exists single-ancilla syndrome extraction circuits that largely preserve the effective distance of the weight-reduced qLDPC codes.
  In addition, we also show that the distance balancing technique in Ref.~\cite{evra2022decodable} preserves effective distance.
  As a corollary, our result shows that higher-dimensional hypergraph product (HGP) codes, also known as homological product codes corresponding to the product of 1-complexes, have no troublesome hook errors when using any single-ancilla syndrome extraction circuit.
\end{abstract}
\tableofcontents

\section{Introduction}\label{sec:intro}

Quantum error correction takes a prominent role in the realization of quantum computation.
Stabilizer codes are a class of quantum error-correcting codes that are widely used in quantum error correction because its underlying theory simplifies error correction \cite{gottesman1997stabilizer}.
Calderbank-Shor-Steane (CSS) codes are a subclass of stabilizer codes that are particularly useful because they can be constructed from classical codes, giving us a way to combine our deep knowledge of classical codes with quantum error correction \cite{calderbank1996good,steane1996multiple}.
Surface codes, for example, are a useful class of CSS code candidate with a high threshold, low-density parity checks (LDPC), and geometrically local connectivity \cite{kitaev2003fault,dennis2002topological,fowler2012surface}.

While surface codes are a leading candidate, they have a vanishing rate.
Thus, researchers hope to find quantum LDPC (qLDPC) codes with good rate which can be practically implemented.
Unfortunately, most known qLDPC codes with linear distance and constant rate have a prohibitively large (though constant) stabilizer weights \cite{panteleev2022asymptotically, leverrier2022quantum}.
Considering the physical implementation of these codes, the large constant weight of the stabilizer generators not only puts them effectively out of reach of near-term physical devices but also leads to expensive error correction circuits with a large number of gates and long syndrome extraction time.
When we consider faults on the ancilla qubits used for syndrome extraction, we may have a substantial decrease in the distance of the code due to the presence of troublesome hook errors~\cite{dennis2002topological, landahl2011fault, bravyi2012subsystem, gidney2023pair}.
This distance---obtained by considering faults on data and ancilla qubits---is known as the \emph{effective distance} of the code~\cite{brown2020critical, bombin2023logical,manes2023distance,grans2024improved}.

Because reducing the weight of the stabilizer generators can lead to a more efficient error correction scheme, several ideas have been put forward to achieve this.
One such strategy involves transforming our CSS stabilizer codes into subsystem codes and then measuring the low constant weight gauge operators to effectively perform stabilizer measurement.
Notable examples include the Bacon-Shor code where weight-4 checks are replaced with weight-2 checks \cite{bacon2006operator, aliferis2007subsystem}, other topological subsystem codes~\cite{bombin2010topological,higgott2021subsystem}, and recent innovations such as Floquet codes~\cite{hastings2021dynamically}.
Other strategies that are more modern involve utilizing the space-time picture of error correction~\cite{bacon2017sparse,gottesman2022opportunities, delfosse2023spacetime} to consider ways that may (in)directly reduce the weight of the stabilizer generators.
Sabo \emph{et al.} have also considered applying weight reduction on the classical codes that are used to construct HGP codes and Lifted Product (LP) codes to achieve weight-reduced HGP and LP codes~\cite{sabo2024weight}.
However, it is not immediately obvious that their strategy would work equally well for other CSS stabilizer qLDPC codes.
In fact, it is unclear how their strategy would work for quantum codes that are not constructed from classical codes. 

In Refs.~\cite{hastings2016weight, hastings2021quantum}, Hastings introduced four techniques to reduce the weight of the stabilizer generators of a CSS stabilizer code.
These techniques are \emph{copying}, \emph{gauging}, \emph{thickening and choosing height}, and \emph{coning}.
Interestingly, Hastings' weight reduction techniques have been generalized for balancing the distance of quantum codes \cite{evra2022decodable} and are relevant in the construction and modification of quantum Locally Testable Codes (qLTCs) with implications in Hamiltonian complexity~\cite{aharonov2015quantum,eldar2017local,wills2023general,wills2023tradeoff,dinur2024expansion}.
While these techniques are successful in reducing the weight of all stabilizer generators to at most 5, they come at the cost of increasing the number of physical qubits by a constant factor that is widely considered to be prohibitive by most near-term standards.
Nonetheless, they are still interesting from a theoretical perspective and may be useful in the future as a tool when qubits are more plentiful and stabilizer weights become a bottleneck.
Most importantly, they have inspired works such as Ref.~\cite{sabo2024weight} to makes progress towards more practical weight reduction.

Remarkably, we find that Hastings's weight reduction techniques have nice fault tolerance properties as they largely preserve the pre-weight reduced codes' scaling of effective distance and sometimes \emph{improving} the effective distance. 
Our work provides evidence that the weight-reduced quantum codes remain relatively fault-tolerant when compared to their pre-weight reduced counterparts.
Moreover, the preservation and improvement of effective distances may be independently interesting from a theory perspective.

In this paper, we analyze how Hastings's weight reduction techniques interacts with the structure of the original quantum code and how these interactions affect the effective distances of the code.
In particular, we focus on \emph{copying}, \emph{gauging}, \emph{thickening and choosing heights}, and \emph{coning}.
Specifically, we study whether there are measurement schedules that can help us avoid troublesome hook errors which are errors within stabilizer measurement schedules which cause multiple data errors.
Our main results are stated informally below.  
\begin{theorem}[Informal Statement of Effective Distance Preservation for Copied and Gauged Codes]\label{thm:main_2_informal}
  Suppose we are given a CSS stabilizer qLDPC code with some single-ancilla syndrome extraction schedule that results in effective distances $d_X$ and $d_Z$.
  There exists some single-ancilla syndrome extraction schedule for the copied and gauged qLDPC code that results in effective distances at least $\Omega\left(1/w_X\right)d_X$ and $d_Z$ where $w_X$ is maximum weight of any $X$ stabilizer generator.
\end{theorem} 

\begin{theorem}[Informal Statement of Effective Distance Preservation for Thickened and Height-Chosen Codes]\label{thm:main_3_informal}
  Suppose we are given a CSS stabilizer qLDPC code with some single-ancilla syndrome extraction schedule that results in effective distances $d_X$ and $d_Z$ and a classical code with distance $d_c$.
  There exists some single-ancilla syndrome extraction schedule for the thickened and height-chosen qLDPC code that results in effective distances $d_c \cdot d_X$ and $d_Z$.
\end{theorem} 

\begin{theorem}[Informal Statement of Distance-Balanced Codes]\label{cor:main_1_informal}
  The distance gained from the distance balancing technique in Ref.~\cite{evra2022decodable} is preserved for any stabilizer measurement schedule.
\end{theorem}

\begin{theorem}[Informal Statement of Effective Distance Preservation for Coned Codes]\label{thm:main_4_informal}
  Suppose we are given a CSS stabilizer qLDPC code with some single-ancilla syndrome extraction schedule that results in effective distances $d_X$ and $d_Z$.
  There exists some single-ancilla syndrome extraction schedule for the coned qLDPC code that results in effective distances at least $d_X$ and $\Omega\left(1/q_X\right) d_Z \lambda \ell$ where $q_X$ is the maximum number of $X$ stabilizer generators that act on any qubit for the original qLDPC code, $\lambda$ is some soundness factor we define later, and $\ell$ is the length of the classical repetition code we use in coning.
\end{theorem}

Surprisingly, we find that as a result of Theorem~\ref{cor:main_1_informal}, higher-dimensional hypergraph product (HGP) codes \cite{zeng2019higher} have an effective distance equal to their code-theoretic distance under any single ancilla qubit measurement schedule.
Our result generalizes the main theorem in Ref.~\cite{manes2023distance} which showed that any stabilizer measurement schedule preserves the distance of HGP codes. 
Given that higher dimensional HGP codes have been shown to afford efficient logical gates~\cite{xu2024fast} and single-shot decoding~\cite{campbell2019theory}, our results strengthen the case for higher dimensional HGP codes to be one of the leading candidates for fault-tolerant quantum computation.
Because higher dimensional HGP codes are a superset of $n$-dimensional toric codes \cite{dennis2002topological, ThreeDToric2005, Jochym_O_Connor_2021} and contain most homological product codes seen in existing literature \cite{freedman2013quantum, bravyi2014homological, quintavalle2021single}, this result may also be of independent interest.
\begin{theorem}[Informal Statement of Hook Error Absence for Higher Dimensional HGP Codes]\label{cor:main_2_informal}
  Higher-dimensional hypergraph product (HGP) codes, also known as homological product codes corresponding to products of 1-complexes, have no troublesome hook errors when using any single-ancilla syndrome extraction circuit.
\end{theorem}

The rest of the paper is organized as follows.
In Section~\ref{sec:notation}, we introduce the notation and definitions that we will use throughout the paper.
In Section~\ref{sec:quantum_weight_reduction}, we introduce the different weight reduction techniques.
In Section~\ref{sec:distance_preservation}, we analyze the effective distance preservation of the weight reduced codes and present our main results.
Lastly, in Section~\ref{sec:discussion}, we discuss the implications of our results and suggest directions for future work.

\section{Background}\label{sec:notation}
\subsection{Classical and Quantum Codes}
An $[n, k, d]$ code is a linear code with $n$ bits, a $k$-dimensional logical space, and distance $d$.
The vector space of codewords, $\calC$, equals $\set{v \in \F_2^n \mid Hv  = 0}$ for a full-rank parity check matrix $H \in \F_2^{(n - k) \times n}$.
In other words, all codewords are orthogonal to the rows of $H$ and lie in the kernel of $H$.
Note that we assume that the parity check matrix $H$ is always full-rank.
We often use $\calC$ to denote the code itself and sometimes include $H$ in the notation, i.e., $\calC(H)$ to refer to the classical code.

An $\left[\left[n, k, d\right]\right]$ (or $\left[\left[n, k, d_X, d_Z\right]\right]$) code is a quantum stabilizer code with $n$ qubits that encode $k$ logical qubits and have distance $d = \min\left(d_X, d_Z\right)$ where $d_X$ and $d_Z$ refer to the minimum weight of non-trivial $X$ and $Z$ logical operators. 
In our work, we focus on CSS codes which can be parameterized by two parity check matrices representing the stabilizer
generators of the CSS code, $H_X$ for the $X$ stabilizer generators and $H_Z$ for the $Z$ stabilizer generators.
Throughout this work, we will provide procedures which modify a CSS code. Unless otherwise stated, the input 
CSS code to each procedure appears as $H_X$ and $H_Z$ and the output appear as $H'_X$ and $H'_Z$.
Let $n_X$ and $n_Z$ be the number of $X$ and $Z$ stabilizers, respectively, $w_X$ and $w_Z$ be the maximum Hamming weight of the $X$ and $Z$ stabilizer generators, respectively, and $q_X$ and $q_Z$ be the maximum Hamming weight of the columns of $H_X$ and $H_Z$, respectively. 
In other words, $q_X$ and $q_Z$ are the maximum number of $X$ and $Z$ stabilizers that any qubit participates in.
Likewise, $w_X$ and $w_Z$ are the maximum number of qubits that any $X$ and $Z$ stabilizer generator acts on respectively.
Because we are mainly interested in qLDPC codes, we assume that $q_X$ and $q_Z$ as well as $w_X$ and $w_Z$ are constants that do not scale with the code size\footnote{We note that the results do generalize to non-constant weights but would result in overhead that scales with the code size.}.
We often use $\calQ$ or $\calQ\left(H_X, H_Z\right)$ to denote the quantum code and $\calQ'$ or $\calQ'\left(H'_X, H'_Z\right)$ to denote the weight-reduced quantum code.
The number of logical qubits $k$ in a quantum code $\calQ$ can be obtained by evaluating the dimension of the quotient group $\ker H_Z / \rs H_X$ or $\ker H_X / \rs H_Z$.
The coset representatives for the two quotient groups are also the $X$ and $Z$ logical operators of $\calQ$ respectively.  

\subsection{Homology and Chain Complexes}
We define an $\ell$-term chain complex $\calA$ (which we will refer to as a chain complex), to be a sequence of vector spaces $\set{A_i}$ over $\F_2$ and $i \in \set{1, \dots, \ell}$.
We define the $i$-th boundary operator as a map $\partial_i: A_{i} \rightarrow A_{i - 1}$ such that $\partial_{i-1} \partial_{i } = 0$.
We can also illustrate a chain complex as:
\[
	\F_2^{\dim A_{i + 1}} \xrightarrow{\partial_{i + 1}} \F_2^{\dim A_i} \xrightarrow{\partial_{i}} \F_2^{\dim A_{i - 1}}.
\]

Moreover, we can define a coboundary operator as the dual of the boundary operator: 
\[\delta_i:~A_{i} \rightarrow A_{i+1}.\]
If we define $\partial_i = L$ where $L$ is a linear operator in some basis (a matrix), then $\delta_{i-1} = L^\top$.
Because $\partial_{i} \partial_{i+1} = 0$ and $\delta_{i} \delta_{i-1} = 0$, we also have that $\im(\partial_{i + 1}) \subseteq \ker(\partial_i)$.
Finally, we define the $i$-th homology group to be $\calH_i(\calA) = \ker(\partial_i) / \im(\partial_{i + 1})$ and the cohomology group be $\calH^i(\calA) = \ker(\delta_{i}) / \im(\delta_{i-1})$.

In general, we denote chain complexes with calligraphic font (i.e.\ $\calA, \calB$) and refrain from using the characters $\calC$ and $\calQ$ to avoid confusion.
For a specific chain complex $\calA$, we will use superscripts ($\partial^\calA$, $\delta^\calA$) to refer to the (co)boundary operators of $\calA$.

Given that $\im(\partial_{i + 1}) \subseteq \ker(\partial_i)$, we have a nice correspondence between CSS codes and chain complexes.
Specifically, if we associate the linear map $H_Z$ to $\delta_1$ and $H_X$ to $\partial_1$, then  $H_X H_Z^\top = \partial_1 \partial_2 = 0$ as necessary.
So, for a quantum code $\calQ(H_X, H_Z)$, we have that $2$-cells are $Z$ stabilizers, $1$-cells are qubits, and $0$-cells are $X$ stabilizers, illustrated by:
\[
	\F_2^{n_Z}	\xrightarrow{H_Z^\top} \F_2^n \xrightarrow{H_X} \F_2^{n_X}.
\]
Similarly, we have the dual
\[
	\F_2^{n_Z} \xleftarrow{H_Z} \F_2^n \xleftarrow{H_X^\top} \F_2^{n_X}.
\]

The $Z$-logicals correspond to $\calH_1(\calQ) = \ker H_X / \im H_Z^\top = \ker H_X / \rs H_Z$ and $X$-logicals to $\calH^1(\calQ) = \ker H_Z / \im H_X^\top = \ker H_Z / \rs H_X$.

We define $\row\left(H \right)$ as the set of rows of some matrix $H$.
In addition, let $\hat{i}_{v}$ be a unit vector in $\F_2^{v}$ with a single 1 in the $i$\textsuperscript{th} index and $\vec{1}_{v}$ be a vector of ones in $\F_2^{v}$.
Let $[n]$ denote the following set of positive integers, i.e., $[n] = \set{1, \ldots, n}$.

\section{Quantum Weight Reduction}\label{sec:quantum_weight_reduction}
In this section, we revisit Hastings's quantum weight reduction techniques proposed in Ref.~\cite{hastings2021quantum} to provide additional helpful exposition on the techniques and standardize notation. 
For a more detailed retelling of Hastings's techniques, we highly recommend looking at Section~3 of Ref.~\cite{sabo2024weight} from which we heavily draw in this section.

The four steps of Hastings's quantum weight reduction method are:
\begin{enumerate}
	\item Copying -- reduces $q_X$
	\item Gauging -- reduces $w_X$
	\item Thickening and choosing heights --reduces $q_Z$
	\item Coning -- reduces $w_Z$
\end{enumerate}
It is easy to see that the combination of these four different techniques would allow us to reduce the four important parameters that quantifies the relevant weights of a quantum code.

\subsection{Copying}\label{sec:copying}

In Ref.~\cite{hastings2021quantum}, copying is introduced as a technique to reduce $q_X$ to at most three.  
In this section, we begin by providing an intuitive picture for how copying transforms the quantum code before providing the details using linear algebraic language and substantiating the exposition with examples from Ref.~\cite{sabo2024weight}.

A copied quantum code $\calQ'$ can be understood as the concatenation of each original qubit of the quantum code $\calQ$ with a classical repetition code $\left[q_X, 1, q_X\right]$.
In other words, we make $q_X$ copies of each original qubit. Then the $q_X$ $X$ stabilizer generators that act on a single original qubit are rearranged to act on one of the $q_X$ copied qubits.
However, notice that the number of logical qubits would increase if we only add these copied qubits without adding new linearly independent stabilizer generators.
Thus, we add $q_X - 1$ new $X$ stabilizer generators for each original qubit that function like the checks for a $\left[q_X, 1, q_X\right]$ classical repetition code as they act on two adjacent qubits in the set of $q_X$ copied qubits that correspond to a particular original qubit.
To protect these copied qubits from Pauli $X$ errors, we then extend the support of every $Z$ stabilizer generator that acts on a particular original qubit to all $q_X$ qubits that correspond to the particular original qubit.
It is easy to see that the stabilizer generators commute with each other and each copied qubit is acted upon by at most 1 original $X$ stabilizer generator and at most 2 new $X$ stabilizer generators.
We provide a diagrammatic representation of the copying technique in Figs.~\ref{fig:bef_copying}~and~\ref{fig:post_copying}.

\pgfdeclarelayer{bg}    
\pgfsetlayers{bg,main}  
\usetikzlibrary{backgrounds}

\begin{figure}[ht]
    \centering
    \begin{subfigure}[b]{0.35\textwidth}
        \centering
        \begin{tikzpicture}[scale=1, every node/.style={draw, inner sep=2pt}]
    
    \foreach \i in {1,2,3,4} {
        \node[blue] (v\i) at (0, -\i) {$s_{X_\i}$};  
    }
    
    \foreach \i in {1} {
        \node[red, circle] (h\i) at (3, -2.5) {$q_{\i}$};  
    }
    
    \foreach \i in {1,2} {
        \node[black] (c\i) at (6, -\i-1) {$s_{Z_\i}$};  
    }
    
    \draw (v1) -- (h1);
    \draw (v2) -- (h1);
    \draw (v3) -- (h1);
    \draw (v4) -- (h1);

    \draw (h1) -- (c1);
    \draw (h1) -- (c2);

\end{tikzpicture}
\caption{Before copying.\label{fig:bef_copying}}
\end{subfigure}
    \hfill
\begin{subfigure}[b]{0.35\textwidth}
        \centering
        \begin{tikzpicture}[scale=1, every node/.style={draw, inner sep=2pt}]
    
    \foreach \i in {1,2,3,4} {
        \node[blue] (v\i) at (0, -\i) {$s_{X_\i}$};  
    }

    \foreach \i in {5,6,7} {
        \node[blue] (v\i) at (1.5, -\i+4-0.5) {$s_{X_\i}$};  
    }
    
    \foreach \i in {1, 2, 3, 4} {
        \node[red, circle] (h\i) at (3, -\i) {$q_{1,\i}$};  
    }
    
    \foreach \i in {1,2} {
        \node[black] (c\i) at (6, -\i-1) {$s_{Z_\i}$};  
    }
    
    \draw (v1) -- (h1);
    \draw (v2) -- (h2);
    \draw (v3) -- (h3);
    \draw (v4) -- (h4);
    \draw (v5) -- (h1);
    \draw (v5) -- (h2);
    \draw (v6) -- (h2);
    \draw (v6) -- (h3);
    \draw (v7) -- (h3);
    \draw (v7) -- (h4);

    \draw (h1) -- (c1);
    \draw (h2) -- (c1);
    \draw (h3) -- (c1);
    \draw (h4) -- (c1);
    \draw (h1) -- (c2);
    \draw (h2) -- (c2);
    \draw (h3) -- (c2);
    \draw (h4) -- (c2);

\end{tikzpicture}
\caption{After copying.\label{fig:post_copying}}
\end{subfigure}
\\
\begin{subfigure}[b]{0.35\textwidth}
        \centering
        \begin{tikzpicture}[scale=1, every node/.style={draw, inner sep=2pt}]
    
    \foreach \i in {1} {
        \node[blue] (v\i) at (0, -2.5) {$s_{X_\i}$};  
    }
    
    \foreach \i in {1, 2, 3, 4} {
        \node[red, circle] (h\i) at (3, -\i) {$q_{\i}$};  
    }
    
    \foreach \i in {1,2} {
        \node[black] (c\i) at (6, -\i-1) {$s_{Z_\i}$};  
    }
    
    \draw (v1) -- (h1);
    \draw (v1) -- (h2);
    \draw (v1) -- (h3);
    \draw (v1) -- (h4);

    \draw (h1) -- (c1);
    \draw (h2) -- (c1);
    \draw (h3) -- (c1);
    \draw (h4) -- (c1);
    \draw (h1) -- (c2);
    \draw (h2) -- (c2);
    \draw (h3) -- (c2);
    \draw (h4) -- (c2);

\end{tikzpicture}
\caption{Before gauging.\label{fig:bef_gauging}}
\end{subfigure}
\hfill
\begin{subfigure}[b]{0.35\textwidth}
        \centering
        \begin{tikzpicture}[scale=1, every node/.style={draw, inner sep=2pt}]
    
    \foreach \i in {1,2,3,4} {
        \node[blue] (v\i) at (0, -\i) {$s_{X_{1,\i}}$};  
    }

    \foreach \i in {1,2} {
        \node[black] (c\i) at (6, -\i-1) {$s_{Z_\i}$};  
    }

    \foreach \i in {5,6,7} {
        \node[red, circle] (h\i) at (2.2, -\i+4-0.5) {$q_{\i}$};  
    }

    \foreach \i in {1, 2, 3, 4} {
        \node[red, circle] (h\i) at (3, -\i) {$q_{\i}$};  
    }
    
    \draw (v1) -- (h1);
    \draw (v2) -- (h2);
    \draw (v3) -- (h3);
    \draw (v4) -- (h4);

    \draw (v1) -- (h5);
    \draw (v2) -- (h5);
    \draw (v2) -- (h6);
    \draw (v3) -- (h6);
    \draw (v3) -- (h7);
    \draw (v4) -- (h7);

    \draw (h1) -- (c1);
    \draw (h2) -- (c1);
    \draw (h3) -- (c1);
    \draw (h4) -- (c1);
  
    \draw (h1) -- (c2);
    \draw (h2) -- (c2);
    \draw (h3) -- (c2);
    \draw (h4) -- (c2);

\begin{scope}[on background layer]
    \draw (h5) -- (c1);
    \draw (h7) -- (c1);
    \draw (h5) -- (c2);
    \draw (h7) -- (c2);
\end{scope}
\end{tikzpicture}
\caption{After gauging.\label{fig:post_gauging}}
\end{subfigure}
\caption{Examples for copying and gauging. (a) The qubits $q_1$ lies in the support of 4 different $X$ stabilizer generators $s_{X_1}, \ldots, s_{X_4}$.
(b) After performing copying, we obtain 4 copies of $q_1$ which are denoted as $q_{1, 1} \ldots, q_{1,4}$. These copied qubits are connected by 3 new $X$ stabilizer generators in a repetition code-like structure. These copied qubits also lie in the same support of the $Z$ stabilizer generators $s_{Z_1}$ and $s_{Z_2}$.
(c) Before gauging, the $s_{X_1}$ has support on 4 different qubits $q_1, \ldots, q_4$.
(d) After gauging, we split the $X$ stabilizer generator into 4 copied $X$ stabilizer generators $s_{X_{1,1}}, \ldots, s_{X_{1,4}}$ which are connected by 3 new qubits $q_5, q_6, q_7$ in a repetition code-like structure. The $Z$ stabilizer generators $s_{Z_1}$ and $s_{Z_2}$ are also updated to have support on the new qubits to ensure that they commute with the copied $X$ stabilizer generators.}
\end{figure}

\subsubsection*{Formal Description}
Firstly, we make $q_X$ copies of each qubit:
\[
	(q_1 \ q_2 \ \dots q_n) \rightarrow (q_{1,1} \dots q_{1,q_X} \ |\  q_{2,1} \dots q_{2,q_X} \ | \ \dots \  |\ q_{n,1} \dots q_{n,q_X})
\]
where $q_i$ represents one of the $n$ original qubits and the partitions are meant to demarcate sets of $q_X$ copied qubits that correspond to some original qubit.
Because we now have $n \cdot q_X$ qubits in our copied code, we have to reshape $H_X$ such that each row that corresponds to an $X$ stabilizer generator now has $n \cdot q_X$ entries instead of $n$. 
In particular, our new $H'_X$ should be constructed such that for every $q_i$ in the support of some row of $H_X$, one of $\{q_{i,1},\cdots,q_{i,q_X} \}$ receives the value of $q_i$ in the corresponding row in $H'_X$.

In order to have effective weight reduction, we must be careful in our choice of $\set{q_{i, 1}, \dots, q_{i, q_X}}$ for every stabilizer. 
Specifically, for two different stabilizer $s_x$ and $s_x'$ from the original code, we require that they do not \emph{share} a qubit $q_{i, j}$ in the copied code.
So,  if \[(1\ 0 \ 0\ | \ 1\ 0 \ 0\ | \ 1\ 0 \ 0)\] is chosen to be the modified version of a stabilizer generator $s_x = (1, 1, 1)$ (a row in $H_X$), then another stabilizer generator $s_x' = (1, 0, 1)$ cannot be set to 
\[
	(0\ 1 \ 0\ | \ 0\ 0 \ 0\ | \ 1\ 0 \ 0)
\]
because $q_{3,1}$ would be acted upon by two different $X$ stabilizer generators.
Fortunately, there is a simple greedy algorithm for updating each row without any collisions. We simply update each $X$ stabilizer one by one and choose the $q_{i, j}$ without support from another $X$ stabilizer in the original code.
This method of construction ensures that each copied qubit is acted upon by at most 1 modified $X$ stabilizer generator.

In addition to the modified $X$ stabilizer generators,  $\left(q_X-1\right)$ new $X$ stabilizer generators are added for each set of $q_X$ qubits to link the copies of $q_i$ such that they collectively ``behave'' like the original single qubit.  
These new $X$ stabilizer generators act on copied qubits $q_{i,j}$ and $q_{i,j+1}$ for $1\leq j \leq q_X-1$. 

In order to make the $Z$ stabilizer generators commutes with all the modified and new $X$ stabilizer generators,  the value at $q_i$ is replicated for $q_{i,j}$ for all $1\leq j \leq q_X$ when we expand the each row in the $H_Z$ to now have $n\cdot q_X$ entries instead of $n$ entries. 
Note that this increase $w_Z$ but it will be addressed with the coning technique.

\begin{lemma}[{\cite[Code Parameters of Copied Quantum Code]{hastings2021quantum}}]
	Suppose we have a quantum CSS code $\calQ$ that encodes $k$ logical qubits with $n$ physical qubits.
	After applying copying on $\calQ$ to obtain $\calQ'$, the parameters of $\calQ'$ are as follows:
	\begin{align*}
  &n' = q_X \cdot n,\\
  &k' = k,\\
  &n'_X = n_X + \left(q_X - 1\right) \cdot n, \\
  &n'_Z = n_Z,\\
  &w'_X = w_X,\\
  &q'_X = \min\left(q_X, 3\right),\\
  &w'_Z = q_X \cdot w_Z,\\
  &q'_Z = q_Z,\\
  &d'_X = d_X,\\
  &d'_Z = q_X \cdot d_Z.
	\end{align*}
\end{lemma}

\subsection{Gauging}\label{sec:gauging}
Hastings proposed gauging to reduce $w_X$ to less than or equal to three without increasing $q_X$.  
Again, we begin by providing a high-level description of gauging to provide the reader with some intuition before providing the details with linear algebraic formalism.

A gauged quantum code $\calQ'$ is very similar to a copied quantum code. Now, instead of ``concatenating'' with a repetition code, we ``concatenate'' within $X$ stabilizer generators in a sort of ``dual'' concatenation. 

Since the goal for gauging is to reduce $w_X$, we make $w_X$ copies of the original $X$ stabilizer generators so that each copy only acts on 1 qubit instead of $w_X$ qubits.
Because we are introducing $w_X - 1$ additional $X$ stabilizers per original $X$ stabilizer generator, we have to introduce $w_X - 1$ new qubits for each original $X$ stabilizer generator to ensure that the number of logical qubits does not change.
We want to assemble each set of $w_X$ copied stabilizers and $w_X - 1$ new qubits such that each of the $w_X - 1$ qubits lies in the support of only two copied $X$ stabilizer generators.
This construction gives us some sort of classical repetition code $\left[w_X, 1, w_X\right]$ where we have the $w_X$ copied $X$ stabilizer generators acting like the bits of the classical repetition code and the $w_X - 1$ new qubits acting like the checks of the classical repetition code.
Notice that our copied $X$ stabilizer generators might no longer commute with the $Z$ stabilizer generators since their supports overlap on at most 1 qubit.
To resolve this, we have to modify our $Z$ stabilizer generators such that they act on some subset of the new qubits introduced during gauging so that the stabilizer generators commute with each other.
With this construction, it is easy to see that each copied $X$ stabilizer generator acts on at most 3 qubits.
We provide a diagrammatic representation of the gauging technique in Figs.~\ref{fig:bef_gauging}~and~\ref{fig:post_gauging}.

\subsubsection*{Formal Description}
For an $X$ stabilizer generator with weight $w > 3$ with supports on qubits $\set{q_1, \dots, q_w}$, gauging adds $w-1$ new qubits labeled as $\set{q_1', \dots, q_{w - 1}'}$. The $X$ stabilizer generator is then updated to a collection of $w$ stabilizer generators:
\[
	\{q_1,q_1' \}, \{q_1', q_2, q_2'\}, \{q_2', q_3, q_3'\}, \cdots, \{q_{w-2}',q_{w-1},q_{w-1}' \}, \{q_{w - 1}', q_w\}.
\]

As pointed out by \cite{sabo2024weight}, we can view the above in matrix form. Say that we have a simple stabilizer, $(1, 1, 1, 1)$, then

\[
	\begin{blockarray}{cccc}
		q_1 & q_2 & q_3 & q_4 \\
		\begin{block}{(cccc)}
			1 & 1 & 1 & 1 \\	
		\end{block}
	\end{blockarray}
	\mapsto
	\begin{blockarray}{cccccccccc}
		& q_1 & q_2 & q_3 & q_4 & & q_1' & q_2' & q_3' & \\
		\begin{block}{(cccccccccc)}
			& 1 &   &   &   & \dots & 1 &   &   & \\
			&   & 1 &   &   & \dots & 1 & 1 &   & \\
			&   &   & 1 &   & \dots &   & 1 & 1 & \\
			&   &   &   & 1 & \dots &   &   & 1 &  \\
		\end{block}
	\end{blockarray}
\]

The commutativity with the $Z$ stabilizers is maintained with a slightly more complicated procedure:
consider the $w$ new rows of a reduced $X$ stabilizers and label them $s_1, \dots, s_w$, giving them an arbitrary order.
If the $j$-th $Z$ stabilizer generator anti-commutes with a string of stabilizers, $s_1 \dots s_i$ for $i \in [w]$, then we add qubit $q_w'$ to the stabilizer generator. Repeat this process until the $Z$ stabilizer generator commutes.

\begin{lemma}[{\cite[Code Parameters of Copied and Gauged Quantum Code]{hastings2021quantum}}]
	Suppose we have a quantum CSS code $\calQ$ that encodes $k$ logical qubits with $n$ physical qubits.
	After applying copying and gauging on $\calQ$ to obtain $\tilde{\calQ}$, the parameters of $\tilde{\calQ}$ are as follows:
	\begin{align*}
  &n' = q_X \cdot n + n_X \left(w_X - 1\right),\\
  &k' = k,\\
  &n'_X = n_X \cdot w_X + \left(q_X - 1\right)\cdot n,\\
  &n'_Z = n_Z,\\
  &w'_X = \min\left(w_X, 3\right),\\
  &q'_X = \min\left(q_X, 3\right),\\
  &w'_Z \leq w_Z \cdot q_X \cdot \left(w_X + 1\right) ,\\
  &q'_Z = q_Z \cdot w_X,\\
  &d'_X \geq \frac{d_X}{w_X/2 + 1},\\
  &d'_Z \geq  q_X\cdot d_Z.
	\end{align*}
\end{lemma}

\begin{remark}
	We have presented the simplified version of \emph{gauging} as outlined in Ref.~\cite{hastings2021quantum}. Ref.~\cite{sabo2024weight} modified \emph{gauging} to have an overhead of $w - 2$ new stabilizer generators and $w - 3$ new qubits.
\end{remark}

\subsection{Thickening and Choosing Heights}\label{sec:thickening_choosing_heights}

In Hastings's original construction, thickening is a technique intended to increase $d_X$ , i.e., the minimum weight of an $X$ logical operator, and choosing heights is meant to reduce $q_Z$, i.e., the maximum number of $Z$ stabilizers that act on any qubit. 
Before we revisit the details of the construction of a thickened chain complex, we first provide the high-level intuition behind thickening.

\tdplotsetmaincoords{80}{120} 
\tdplotsetrotatedcoords{0}{0}{0} 

\begin{figure}[ht]
    \centering
    \begin{subfigure}[b]{0.45\textwidth}
        \centering
        \begin{tikzpicture}[scale=1,tdplot_rotated_coords,
                            rotated axis/.style={->,purple,ultra thick},
                            axis/.style={->,blue,ultra thick},
                            blackBall/.style={ball color = black!80},
                            greyBall/.style={ball color = black!20},
                            radiationAlignment/.style={ultra thin, dashed, opacity = 0.6},
                            cube/.style={dashed, ultra thin,fill=red, opacity = 0.4},
                            borderBall/.style={ball color = white,opacity=.25},
                            wave/.style={
                            decorate,decoration={snake,post length=1.4mm,amplitude=1mm,
                            segment length=2mm},red, thick},
                            very thick]
        
        \foreach \x in {1,2} {
            \foreach \y in {-1,0,1} {
                \ifthenelse{\lengthtest{\x pt < 2pt}} {
                    \fill[red, opacity=0.5] (\x,\y,0) -- (\x+1,\y,0) -- (\x+1,\y+1,0) -- (\x,\y+1,0) -- cycle;
                }{}

                \ifthenelse{\lengthtest{\x pt < 2pt} \and \lengthtest{-1pt < \y pt} \and \lengthtest{\y pt < 3pt}} {
                    \draw[thick] (\x,\y,0) -- (\x+1,\y,0);
                }{}
                \draw[thick] (\x,\y,0) -- (\x,\y + 1,0);
                \ifthenelse{\lengthtest{\y pt > -1pt} \and \lengthtest{-1pt < \y pt} \and \lengthtest{\y pt < 3pt}} {
                    \shade[rotated axis,greyBall] (\x,\y,0) circle (0.05cm);
                }{}
            }
        }
        \end{tikzpicture}
        \caption{Regular surface code}
	\label{fig:thicken:a}
    \end{subfigure}
    \hfill
    \begin{subfigure}[b]{0.45\textwidth}
        \centering
        \begin{tikzpicture}[scale=1,tdplot_rotated_coords,
                            rotated axis/.style={->,purple,ultra thick},
                            axis/.style={->,blue,ultra thick},
                            blackBall/.style={ball color = black!80},
                            greyBall/.style={ball color = black!20},
                            radiationAlignment/.style={ultra thin, dashed, opacity = 0.6},
                            cube/.style={dashed, ultra thin,fill=red, opacity = 0.4},
                            borderBall/.style={ball color = white,opacity=.25},
                            wave/.style={
                            decorate,decoration={snake,post length=1.4mm,amplitude=1mm,
                            segment length=2mm},red, thick},
                            very thick]
        
        \foreach \x in {2, 3} {
            \foreach \y in {-1,0,1} {
                \foreach \z in {0,1,2} {
                    \ifthenelse{\lengthtest{\x pt < 3pt}} {
                        \fill[red, opacity=0.5] (\x,\y,\z) -- (\x+1,\y,\z) -- (\x+1,\y+1,\z) -- (\x,\y+1,\z) -- cycle;
                    }{}
                    \ifthenelse{\lengthtest{\z pt < 2pt} \and \lengthtest{\x pt < 3pt}} {
                        \fill[yellow, opacity=0.5] (\x,\y,\z) -- (\x,\y,\z + 1) -- (\x+1,\y,\z + 1) -- (\x + 1,\y,\z) -- cycle;
                    }{}
                    \ifthenelse{\lengthtest{\z pt < 2pt}} {
                        \fill[yellow, opacity=0.5] (\x,\y,\z) -- (\x,\y,\z + 1) -- (\x,\y + 1,\z + 1) -- (\x,\y + 1,\z) -- cycle;
                    }{}

                    \ifthenelse{\lengthtest{\x pt < 3pt} \and \lengthtest{-1pt < \y pt} \and \lengthtest{\y pt < 3pt}} {
                        \draw[thick] (\x,\y,\z) -- (\x + 1,\y,\z);
                    }{}
                    \draw[thick] (\x,\y,\z) -- (\x,\y + 1,\z);
                    \ifthenelse{\lengthtest{\y pt > -1pt} \and \lengthtest{-1pt < \y pt} \and \lengthtest{\y pt < 3pt}} {
                        \shade[rotated axis,greyBall] (\x,\y,\z) circle (0.05cm);
                        \ifthenelse{\lengthtest{\z pt < 2pt} \and \lengthtest{-1pt < \y pt} \and \lengthtest{\y pt < 3pt}} {
                            \draw[thick] (\x,\y,\z) -- (\x,\y, \z + 1);
                        }{}
                    }{}
                }
            }
        }
        \end{tikzpicture}
        \caption{Thickened surface code with $\ell = 3$}
	\label{fig:thicken:b}
    \end{subfigure}
    \vfill
    \begin{subfigure}[b]{0.45\textwidth}
        \centering
        \begin{tikzpicture}[scale=1,tdplot_rotated_coords,
                            rotated axis/.style={->,purple,ultra thick},
                            axis/.style={->,blue,ultra thick},
                            blackBall/.style={ball color = black!80},
                            greyBall/.style={ball color = black!20},
                            radiationAlignment/.style={ultra thin, dashed, opacity = 0.6},
                            cube/.style={dashed, ultra thin,fill=red, opacity = 0.4},
                            borderBall/.style={ball color = white,opacity=.25},
                            wave/.style={
                            decorate,decoration={snake,post length=1.4mm,amplitude=1mm,
                            segment length=2mm},red, thick},
                            very thick]
        
        \foreach \x in {2,3} {
            \foreach \y in {-1,0,1} {
                \foreach \z in {0,1,2} {
                    \ifthenelse{\y=-1 \and \x=2 \and \z=0} {
                        \fill[red, opacity=0.5] (\x,\y,\z) -- (\x+1,\y,\z) -- (\x+1,\y+1,\z) -- (\x,\y+1,\z) -- cycle;
                    }{}
                    \ifthenelse{\y=0 \and \x=2 \and \z=1} {
                        \fill[red, opacity=0.5] (\x,\y,\z) -- (\x+1,\y,\z) -- (\x+1,\y+1,\z) -- (\x,\y+1,\z) -- cycle;
                    }{}
                    \ifthenelse{\y=1 \and \x=2 \and \z=2} {
                        \fill[red, opacity=0.5] (\x,\y,\z) -- (\x+1,\y,\z) -- (\x+1,\y+1,\z) -- (\x,\y+1,\z) -- cycle;
                    }{}

                    \ifthenelse{\lengthtest{\x pt < 3pt} \and \lengthtest{-1pt < \y pt} \and \lengthtest{\y pt < 3pt}} {
                        \draw[thick] (\x,\y,\z) -- (\x+1,\y,\z);
                    }{}
                    \draw[thick] (\x,\y,\z) -- (\x,\y + 1,\z);
                    \ifthenelse{\lengthtest{\y pt > -1pt} \and \lengthtest{-1pt < \y pt} \and \lengthtest{\y pt < 3pt}} {
                        \shade[rotated axis,greyBall] (\x,\y,\z) circle (0.05cm);
                        \ifthenelse{\lengthtest{\z pt < 2pt} \and \lengthtest{-1pt < \y pt} \and \lengthtest{\y pt < 3pt}} {
                            \draw[thick] (\x,\y,\z) -- (\x,\y, \z + 1);
                        }{}
                    }{}
                }
            }
        }
        \end{tikzpicture}
        \caption{Thickened and height chosen surface}
	\label{fig:thicken:c}
    \end{subfigure}
    \hfill
    \caption{As per convention, vertices represent $X$ stabilizer generators, edges are qubits, and highlighted faces are $Z$ stabilizers. In (a) we have a surface code with $d_X = 2$ and $d_Z = 3$. (b) is the surface code but after thickening, i.e., $d_X' = 6$ and $d_Z' = 3$. The red highlighted faces correspond to $Z$ stabilizer generators of the original code ($A_2 \otimes B_0$) and the yellow highlighted faces correspond to $Z$ stabilizer generators in $A_1 \otimes B_1$. (c) is the thickened code after choosing heights. For visibility, we do not highlight the stabilizers in $A_1 \otimes B_1$ as they remain unchanged from thickening. 
    We also note that $q_Z$ does not decrease in this particular example because it was already less than or equal to 3.
    This example is mainly chosen to illustrate the thickening and height-choosing process.
    }
    \label{fig:thickening}
\end{figure}

A thickened code can be understood as the quantum code resulting from the tensor product between the 2-complex of a quantum code with the 1-complex that corresponds to a classical repetition code.
To be more precise, we are actually using the 1-co-complex of the classical repetition code but will refer to it as a 1-complex for simplicity.
For the sake of clarity in our initial exposition, we assume for now that the 2-complex corresponding to the original quantum code can be represented as a 2-dimensional surface code as shown in Figure~\ref{fig:thicken:a}.

After thickening, the 2-complex of the thickened code can be represented in a 3-dimensional cuboid: the original 2-complex is laid down on the bottom layer of the cube and the 1-complex corresponding to the classical repetition code extends the surface in the vertical direction corresponding to the third dimension.
The new $Z$ stabilizers that are introduced in the thickened code are oriented vertically and their addition means that the logical $X$ operators of the original quantum code no longer commute with the new $Z$ stabilizer group.
In order to satisfy the commutation relations with the new $Z$ stabilizers, the logical $X$ operators of the original quantum code must be extended to act on the newly introduced qubits in a way that is reminiscent of that of the 3D toric code, resulting in an increase in $d_X$. We provide a visualization in Figure~\ref{fig:thicken:b}.

Choosing heights is a technique used to resolve the redundancies that arise from thickening.
By only choosing to keep a select few of the $Z$ stabilizers from each of the layers in the third dimension that are created from thickening, we retain the increase in $d_X$ while reducing the number of $Z$ stabilizer generators that act on each qubit as in Figure~\ref{fig:thicken:c}.

\subsubsection*{Formal Description}
Now that we have provided an intuitive understanding of thickening and choosing heights, we now provide a more formal description of the thickening process. 
Suppose we are given a quantum code $\calQ$ that has a 2-complex $\calA = \left(A_2, A_1, A_0\right) \cong \left(\F_2^{n_Z}, \F_2^n, \F_2^{n_X}\right)$ with the following boundary maps:
\begin{equation}
	\begin{aligned}
		\partial_2: A_2 &\xrightarrow{H_Z^\top} A_1,\\
		\partial_1: A_1 &\xrightarrow{H_X} A_0,
	\end{aligned}
\end{equation}
where $n_Z, n, n_X$ are the number of $Z$ stabilizers, qubits, and $X$ stabilizers respectively. 

We also have a 1-complex $\calB = \left(B_1, B_0\right) \cong \left(\F_2^{\ell - 1}, \F_2^{\ell}\right)$ that corresponds to a classical $\left[\ell, 1, \ell\right]$ repetition code with the following boundary map:
\begin{equation}
	\partial_1: B_1 \xrightarrow{H_R^\top} B_0,
	\label{eq:boundaryChain}
\end{equation}
where $\ell$ is the length of the repetition code and $H_R$ is the parity check matrix of the repetition code.

The 3-complex $\calT$ of the thickened code is then given by the tensor product of $\calA$, the 2-complex of the original quantum code, with $\calB$, the 1-complex of the classical repetition code:
\begin{align}
	\calT &= \calA \otimes \calB \\
	&= \left(T_3, T_2 , T_1, T_0\right) \\
	&\cong \left(A_2 \otimes B_1, \left(A_2 \otimes B_0\right) \oplus \left(A_1 \otimes B_1\right), \left(A_1 \otimes B_0\right) \oplus \left(A_0 \otimes B_1\right), \left(A_0 \otimes B_0\right)\right) \\
	&\cong \left(\F_2^{n_Z \ell}, \F_2^{n_Z\ell + n(\ell-1)}, \F_2^{n\ell + n_X(\ell-1)}, \F_2^{n_X \ell}\right)
\end{align}
where $T_2$ corresponds to the vector space spanned by the new set of $Z$ stabilizer generators.
We can similarly interpret $T_1$ and $T_0$ in the same way but for the new set of qubits and the new set of $X$ stabilizer generators respectively.
Note that $T_3$ corresponds to new $Z$ stabilizers that are clearly redundant and can be removed from the code because they do not have any homological meaning with respect to error correction.
It should be clear that the thickened code $\calQ'$ has $n_Z\ell + n(\ell-1)$ $Z$ stabilizers, $n\ell + n_X(\ell-1)$ qubits, and $n_X\ell$ $X$ stabilizers.
We state the thickened code $\calQ'$'s boundary maps below:
\begin{equation}
	\begin{aligned}
		\partial_2: T_2 &\xrightarrow{H_Z^{\prime \top}} T_1,\\
		\partial_1: T_1 &\xrightarrow{H'_X} T_0,
	\end{aligned}
\end{equation}
where $H^{\prime\top}_Z$ and $H'_X$ are the new parity check matrices of the $Z$ and $X$ stabilizers of the thickened code respectively.
In Ref.~\cite{hastings2021quantum, sabo2024weight}, the authors used K\"unneth's formula to provide an explicit construction for the new parity check matrices of the thickened code which is given by
\begin{align}
	H'_Z &= \left(\begin{array}{c|c}
		H_Z \otimes \mathbbm{1}_\ell & \textbf{0} \\
		\hline
		\rule{0pt}{4mm} 
		\mathbbm{1}_n \otimes H_R & H_X^{\top} \otimes \mathbbm{1}_{\ell - 1}
\end{array}\right),\\
			H'_X &= \left(\begin{array}{c|c} H_X \otimes \mathbbm{1}_{\ell} & \mathbbm{1}_{n_X} \otimes H_R^{\top}\end{array}\right)
		\end{align}
		where the left partition of the parity matrices acts on the $\left(A_1 \otimes B_0\right)$ qubits and the right partition acts on the $\left(A_0 \otimes B_1\right)$ qubits.
		For the $H'_Z$ parity check matrix, the top row corresponds to the $Z$ stabilizers from the $\left(A_2 \otimes B_0\right)$ vector space and the bottom row corresponds to the $Z$ stabilizers from the $\left(A_1 \otimes B_1\right)$ vector space.

		We now state a useful theorem from Ref.~\cite{zeng2019higher} that gives us the distance of the thickened code.
		In the theorem, the authors define the distance of a homology group as the minimum Hamming weight of a nontrivial element (any representative) in the homology group.
		\begin{theorem}[{\cite[Restatement of Theorem 1]{zeng2019higher}}]\label{thm:thickening_distance}
			Consider $m$-complex $\calA$ and assume that homological groups $H_j(\calA)$ and cohomological groups $H^j(\calA)$ have distances $d_j$ and $d^j$ respectively for $0 \leq j \leq m$.
			Also, consider 1-complex $\calB$ with an $r \times c$ binary matrix $H$ of rank $r$ as its boundary operator.
			Then, the distance $d_j'$ of the $j$\textsuperscript{th} homology group $\ker\left(\partial_j\right)/ \im\left(\partial_{j+1}\right)$ of the thickened complex $\calA \otimes \calB$ satisfies the following:
			\[d_j'  = \min\left(d_{j-1}(\calA)d_1(\calB), d_j(\calA) d_0(\calB)\right).\]
			Similarly, we have
			\[d^j\,'  = \min\left(d^{j-1}(\calA)d^1(\calB), d^j(\calA) d^0(\calB)\right).\]
		\end{theorem}

		Using the above theorem, we can explicitly compute the distance of the thickened code $\calQ'$ as shown in the following lemma:
		\begin{lemma}\label{lem:thickened_code_distance}
			Let $d_X$ (corr. $d_Z$) be the minimum weight of the $X$ (corr. $Z$) logical operators of the original quantum code $\calQ$ and $\ell$ be the distance of the $[\ell, 1, \ell]$ repetition code.
			Then, the thickened code $\calQ'$ has an $X$ distance $d_X' = d_X\cdot \ell$ and a $Z$ distance $d_Z' = d_Z$.
		\end{lemma}
		\begin{proof}
			We begin by evaluating the $X$ distance of the thickened code. 
			The $X$ logical operators of the thickened code correspond to the first cohomology group of the thickened complex $\calT = \calA \otimes \calB$.
			We can apply Theorem~\ref{thm:thickening_distance} to the first cohomology group of the thickened complex to obtain:
			\[d_X' = d_X\cdot \ell.\]

			The $Z$ distance of the thickened code can be obtained similarly using Theorem~\ref{thm:thickening_distance}, i.e., $d_Z' = d_Z$.
		\end{proof}

		We now provide a formal description of choosing heights. 
		Suppose we have a thickened code $\calQ'$ with $n_Z\ell + n(\ell-1)$ $Z$ stabilizers and $n\ell + n_X(\ell-1)$ qubits. 
		Recall that there are two types of $Z$ stabilizers in the thickened code: the $Z$ stabilizers arising from the original quantum code that corresponds to the vector space $\left(A_2 \otimes B_0\right)$, and the new $Z$ stabilizers that are introduced in the thickening process that correspond to the vector space $\left(A_1 \otimes B_1\right)$.
		Notice that each qubit in $\left(A_1 \otimes B_0\right)$ is in the boundary of at most 2 $Z$-stabilizers in $\left(A_1 \otimes B_1\right)$ and at most $q_Z$ $Z$-stabilizers in $\left(A_2 \otimes B_0\right)$.
		Also, each qubit in $\left(A_0 \otimes B_1\right)$ is in the boundary of at most $q_X$ stabilizers in $\left(A_1 \otimes B_1\right)$ and no stabilizer in $\left(A_2 \otimes B_0\right)$.
		Suppose $q_X$, the maximum number of $X$ stabilizers that act on any qubit in the original code $\calQ$, has already been weight-reduced by other techniques.
		Then, the $Z$ stabilizers in $\left(A_1 \otimes B_1\right)$ have weight at most $2 + q_X$ and would also not be the reason for a high $q_Z'$.

		In Hastings's original construction, choosing heights is a technique applied on the $Z$ stabilizers in the vector space $\left(A_2 \otimes B_0\right)$ to reduce $q_Z'$.
		Instead of keeping all $Z$ stabilizers in $\left(A_2 \otimes B_0\right)$, we only keep a select few of them to reduce $q_Z'$.
		To be concrete, let $b_1, b_2, \ldots, b_{\ell}$ be the basis elements of the vector space $B_0$ that can be interpreted as the different heights in the 3-dimensional cube representation of the complex $\calT$.
		Similarly, let $z_1, z_2, \ldots, z_{n_Z}$ be the basis elements of the vector space $A_2$ that can be interpreted as the different $Z$ stabilizer generators in the original quantum code $\calQ$.
		For each $Z$ stabilizer generator $z_i$, instead of keeping every $Z$ stabilizer in $z_i \otimes B_0$, we choose to keep only the $Z$ stabilizer generator corresponding to $z \otimes b_i$ for some $i \in [\ell]$.
		Thus, the resulting $Z$ stabilizer generators after we have chosen heights would have the form $z_i \otimes b_j$ for some $j \in [\ell]$ for each $i \in [n_Z]$.
		To see how this reduces $q_Z'$, recall that each qubit in $\left(A_1 \otimes b_i\right)$ for some basis element $b_i \in B_0$ is in the boundary of at most $q_Z$ $Z$-stabilizers in $\left(A_2 \otimes B_0\right)$.
		After choosing heights, some of the $q_Z$ $Z$-stabilizers in $\left(A_2 \otimes b_i\right)$ for each $b_i$ may be removed if their heights were not chosen to be $b_i$, potentially resulting in a decrease in $q_Z$.
		For some target $q_Z'$ which we hope to reduce $q_Z$ to, Hastings provided several different ways to compute the value of $\ell$ required for the classical repetition code for the thickening process and the heights to be chosen in his original construction in Ref.~\cite{hastings2021quantum}.

		For the rest of our exposition, we assume that the value of $\ell$ has been chosen to be sufficiently large to reduce $q_Z'$ to $w_X$, i.e., we have chosen $w = w_X - 2$ where $w$ is the number of $Z$ stabilizer generators in $\left(A_2 \otimes B_0\right)$ that acts on any qubit.

		In Ref.~\cite{sabo2024weight}, the authors provide a concise explanation of how choosing heights reduces $q_Z'$. They pointed out how choosing different heights for any two different $Z$ stabilizer generators in $H_Z$ results in the corresponding two rows in $H'_Z$ not sharing a common boundary.
		This means that the two $Z$ stabilizers do not act on the same qubits and thus choosing heights reduces $q_Z'$. 
		The authors also provided an extreme example where choosing our classical repetition code to have $\ell = n_Z$ and heights = $\left(1, \ldots, n_Z\right)$ guarantees that every qubit in $\left(A_1 \otimes B_0\right)$ is acted upon by at most 1 $Z$-stabilizer in $\left(A_2 \otimes B_0\right)$ and at most 2 $Z$-stabilizer in $\left(A_1 \otimes B_1\right)$.
		This once again emphasizes the importance of choosing the right value of $\ell$ and the right heights to reduce $q_Z'$ because it is very much a careful balancing act of reducing $q_Z'$ while not introducing too much overhead in terms of the number of qubits.

		\begin{lemma}[{\cite[Code Parameters of Thickened and Height-Chosen Quantum Code]{hastings2016weight, hastings2021quantum}}]
			For a quantum CSS code $\calQ$ that encodes $k$ logical qubits with $n$ physical qubits,
			after applying thickening and choosing heights on $\calQ$ to obtain $\tilde{\calQ}$, the parameters of $\tilde{\calQ}$ are as follows:
			\begin{align*}
  &n' = \ell \cdot n + n_X \left(\ell - 1\right), \\
  &k' = k,\\
  &n'_X = \ell \cdot n_x,\\
  &n'_Z = n_Z + (\ell - 1) n,\\
  &w'_X = w_X + 2 \text{ if } \ell \geq 3 \text{ and } w_X + 1 \text{ if } \ell = 2,\\
  &q'_X = \max(q_X, 2),\\
  &w'_Z = \max(w_Z, q_X + 2),\\
  &d'_X = \ell \cdot d_X, \\
  &d'_Z = d_Z
			\end{align*}
			and where $q'_Z = O(1)$ and can be set to $3$ depending on the choice of $\ell$.
		\end{lemma}

\subsection{Coning}
\label{sec:coning}

In Hastings's original work, coning is a technique that reduces $w_Z$.
Before we dive into some of the details of coning, we first provide a diagrammatic representation of the coning operation, found in Ref.~\cite{hastings2021quantum}, in Fig.~\ref{fig:coning}.
The high-level idea of coning is to build separate chain complexes $\overline{\calB}_i$ that replaces each high-weight $Z$ stabilizer generator $s_{Z,i}$.
For each of these new chain complexes, the number of $Z$ stabilizer generators in them corresponds to the weight of the original $Z$ stabilizer generator that it replaces.
To be more precise, each of these new $Z$ stabilizer generators can be mapped to a single qubit that $s_{Z, i}$ acts on.
In addition to acting on the single qubits that they are mapped to, these new $Z$ stabilizer generators each acts on a set of new qubits that are introduced in the coning operation.
The new qubits introduced in the coning operation each loosely corresponds to some $X$ stabilizer generator that is in the boundary of some original qubit in the boundary of $s_{Z,i}$.
A newly added $Z$ stabilizer generator in $\overline{\calB}_i$ acts on the original qubit that it is mapped to as well as the new qubits that are mapped to the $X$ stabilizer generators that lie in the boundary of the original qubit that it is mapped to.
To ensure that we do not increase the dimension of the homology group, we include new $X$ stabilizer generators that are connected to the different new qubits in the chain complex.
In other words, each of these new qubits is acted upon by the new $X$ stabilizer generator as well as the original $X$ stabilizer generator that is mapped to the new qubit.
If we were to let our $X$ stabilizer generators be vertices, qubits be edges, and $Z$ stabilizer generators be faces, we can immediately see a cone structure that replaces a face that has many edges in its boundary.
To be explicit, the newly added edges will ``grow'' from the vertices of an original face belonging to a high weight $Z$ stabilizer generator before convening on a single newly added vertex in the middle of the original face.
These new edges effectively cut up the original high weight face into many low weight faces.
The resulting quantum code is what we term the \emph{cone code}.

\begin{figure}[t]
	\centering
	\includegraphics[width=0.55\textwidth]{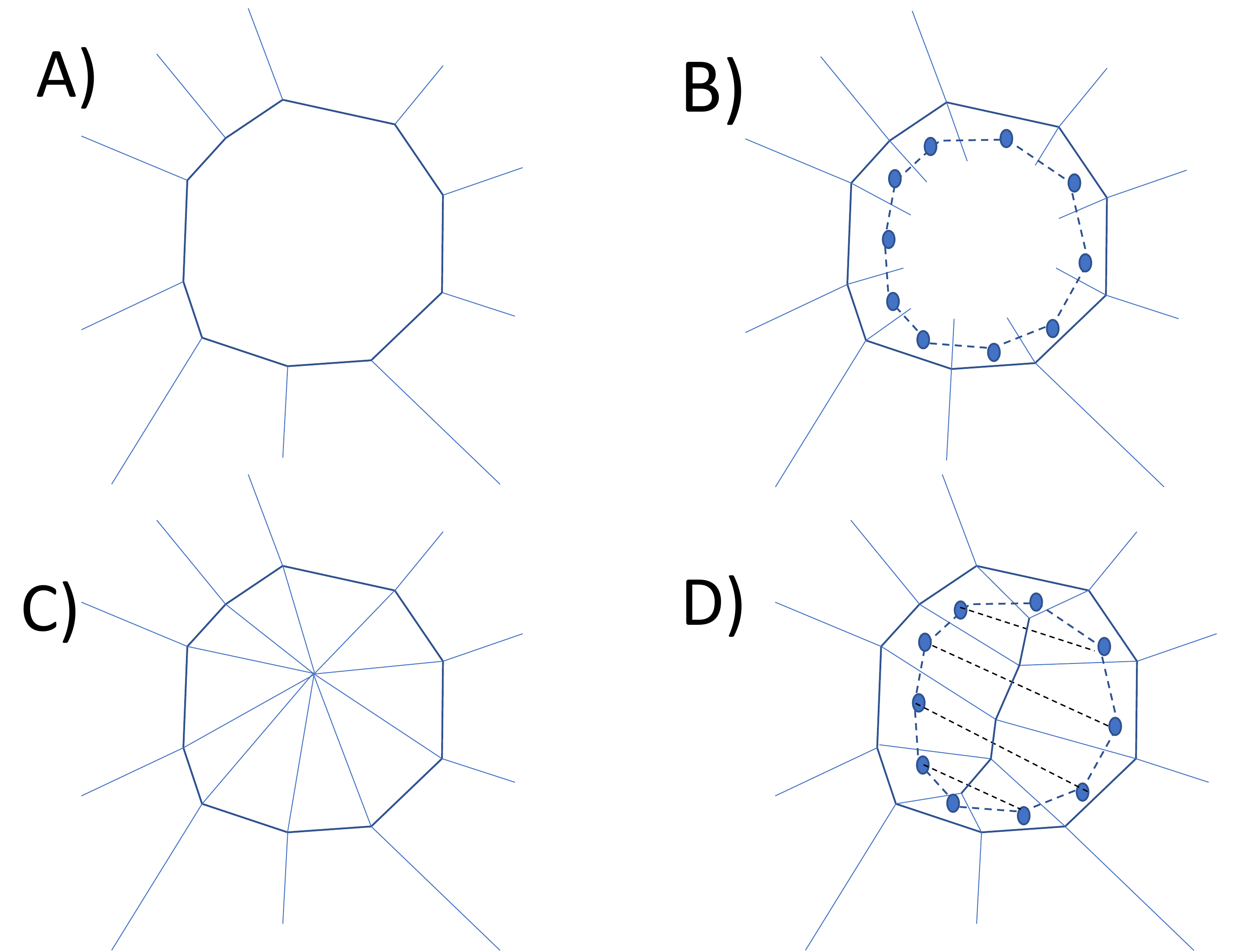}
	\caption{
		A visualization of the coning operation.
		In (A), we start with a high weight $Z$ stabilizer generator where edges correspond to qubits, faces correspond to $Z$ stabilizer generators, and vertices correspond to $X$ stabilizer generators as per convention.
		Then in (B), we build a new chain complex that has chains that correspond to the qubits and $X$ stabilizers that neighbor the high weight $Z$ stabilizer generator.
		Next, in (C), we cone the high weight $Z$ stabilizer generator to obtain the cone code.
		Finally, in (D), we cellulate the coned cone to obtain the reduced cone code.
		This figure is obtained from Ref.~\cite{hastings2021quantum}.
	}
	\label{fig:coning}
\end{figure}

\subsubsection*{Formal Description}
Now we provide a more formal description of the coning operation.
Suppose we are given a quantum code $\calQ$ that has a 2-complex $\calA = \left(A_2, A_1, A_0\right) \cong \left(\F_2^{n_Z}, \F_2^{n}, \F_2^{n_X}\right)$.
Let us define $\calA^\ast$ to be the same complex as $\calA$ but with a smaller set of $Z$ stabilizer generators.
If we let $\calS_Z$ be the set of $Z$ stabilizer generators in $\calQ$, then we can define $\calS_{Z,\calA^\ast}$ be the set of direct $Z$ stabilizer generators, i.e., the set of $Z$ stabilizer generators that we choose not to weight-reduce or ``cone''.
These generators typically have weight lesser or equal to 5 and they are all the $Z$ stabilizer generators in $\calA^\ast$.
Naturally, we can define $\calS_{Z,\calA^\ast}^{c}$ to be the set of $Z$ stabilizer generators that we choose to remove from $\calA$ for coning.
Note that $\calS_{Z, \calA^\ast}^c$ typically contains generators that have support on a large number of qubits and $\calS_Z = \calS_{Z,\calA^\ast} \cup \calS_{Z,\calA^\ast}^{c}$.
Define the chain complex $\calB$ to be a direct sum of 3-term chain complexes $\overline{\calB}_i$ where each $\overline{\calB}_i$ corresponds to a single $s_{Z,i} \in \calS_{Z,\calA^\ast}^{c}$.
Lastly,  let $\calS_{Z,\calB}$ the set of induced $Z$ stabilizer generators in $\calB$ for the cone code that replaces $\calS_{Z, \calA^\ast}^{c}$.

For each of the 3-term complexes $\overline{\calB}_i$, it contains -1-, 0-, and 1-cells.
The set of 1-cells corresponds to the set of qubits that the $Z$ stabilizer generator $s_{Z,i}$ acts on.
We denote the space spanned by the 1-cells as $Q_i = \left(\overline{\calB}_i\right)_1$.
For the set of 0-cells, we need to first define a few things.
Let $\calS_{X,i}$ be the set of $X$ stabilizer generators that are in the boundary of the qubits that are in the boundary of $s_{Z,i}$.
For each $X$ stabilizer generator $s_{X,j} \in \calS_{X,i}$, there exists an even number of qubits associated to $Q_i$ that are in the coboundary of $s_{X,j}$.
Thus, it is possible to pair up these qubits to form a set of tuples $\left(s_{X,j}, q_k, q_l\right)$ where each of these qubits only appear once in a tuple.
The set of 0-cells in $\overline{\calB}_i$ then corresponds to the set of tuples $\left(s_{X, j}, q_k, q_l\right)$ for each $s_{X,j}$ and each pair of qubits $\left(q_k, q_l\right)$ that lie in the support of $s_{X,j}$.
The boundary of a given 1-cell is the set of 0-cells that corresponds to tuples $\left(s_{X,j}, q_k, q_l\right)$ where the given 1-cell corresponds to either $q_k$ or $q_l$.
We denote the space spanned by the 0-cells as $X_i = \left(\overline{\calB}_i\right)_0$.
Because we do not want the $\overline{\calB}_i$ to increase the number of logical qubits in $\calQ$, we introduce -1-cells to make the zeroth homology group trivial.
To be specific, the decongestion lemma from Ref.~\cite{freedman2021building} allows us to find a basis of simple 0-cycles with bounded length for the zeroth homology group of $\overline{\calB}_i$.
For each of these simple 0-cycles, we introduce a -1-cell whose coboundary is the set of 0-cells that form the simple 0-cycle.
We denote the space spanned by the -1-cells as $R_i = \left(\overline{\calB}_i\right)_{-1}$.
Let $\partial_B$ and $\delta_B$ be the boundary and coboundary operators as described above for $\calB = \bigoplus_i \overline{\calB}_i = \left(\bigoplus_i \left(\overline{\calB}_i\right)_1, \bigoplus_i \left(\overline{\calB}_i\right)_0, \bigoplus_i \left(\overline{\calB}_{i}\right)_{-1}\right) = \left(\bigoplus_i Q_i, \bigoplus_i X_i, \bigoplus R_i\right)$.
Let $\partial_{A^\ast}$ and $\delta_{A^\ast}$ be the boundary and coboundary operators for $\calA^\ast$ such that the $Z$ stabilizer generators in $\calS_{Z,\calA^\ast}^c$ do not show up in the coboundary of the qubits in $\calA^\ast$.

Before we formally define what a cone code is, we first define chain maps $f_i:\overline{\calB}_i \to \calA$.
The chain map $f_i$ essentially maps the 1-cells in $\overline{\calB}_i$ to the qubits that they correspond to in $\calA$.
Similarly, it maps the 0-cells in $\overline{\calB}_i$ to the $X$ stabilizer generators that they correspond to in $\calA$.
Lastly, it vanishes on the -1-cells.
We can then define a chain map $f: \calB \to \calA$ as the direct sum of the chain maps $f_i$, i.e., 
\begin{align*}
  f = \begin{pmatrix}
    f_1 \\ f_2 \\ \vdots
  \end{pmatrix}.
\end{align*}
Now, we are ready to define the cone code which we shall denote as $\Cone(\calQ, f)$.
\begin{definition}[Cone Code]
  Given a chain map $f: \calB \to \calA$, the cone of $\calQ$ with respect to $f$ is the quantum code $\Cone(\calQ, f)$ that has a 3-term chain complex such that $\left(\Cone(\calQ, f)\right)_j = A_j \oplus B_{j-1} = A_j \oplus \left(\bigoplus_i \left(\overline{\calB}_i\right)_{j-1}\right)$ for $j \in \{0, 1, 2\}$.
  The boundary operator for $\Cone(\calQ, f)$ is given by
  \begin{align*}
    \partial_{\Cone(\calQ, f)} = \begin{pmatrix}
      \partial_A & f \\
      0 & \partial_B
    \end{pmatrix}
  \end{align*}
\end{definition}

Because a 0-cell in $\overline{\calB}_i$ can potentially be part of $\calO\left(\log^2 n\right)$ many simple 0-cycles in the basis for the zeroth homology group of $\overline{\calB}_i$, we may end up with newly introduced qubits that are acted upon by polylogarithmically many $X$ stabilizer generators.
To ensure that $q_X$ does not blow up because of coning, we perform \emph{thickening and choosing heights} on $\Cone(\calQ, f)$ with the thickening parameter $\ell = \Theta\left(\log\left(w_Z\right)^{2+2\epsilon} w_Z^\epsilon\right)$ for any $\epsilon > 0$.
Note that we are reducing $q_X$ so we are thickening in the dual basis compared to the original thickening operation formulated in the previous section.
We refer to the resulting code as the \emph{thickened and height-chosen cone code} $\Cone(\calQ, f, \ell)$.
Another issue pertains to how a -1-cell in $\overline{\calB}_i$ can potentially lie in the boundary of $\left|Q_i\right|$ many 0-cells because the simple 0-cycle can be of length up to $\left|Q_i\right|$.
This would result in a blow up in $w_X$ if the -1-cell is involved in the boundary of many 0-cells.
To address this, Ref.~\cite{hastings2021quantum} introduces a cellulation procedure that is applied on each of the $\overline{\calB}_i$.
By introducing additional 0-cells and -1-cells in each of the simple cycles that are part of the basis for the zeroth homology group of $\overline{\calB}_i$, we can break the high weight -1-cells into many low weight -1-cells connected by the additional 0-cells to ensure that $w_X$ is a small constant.
The resulting code is what we term the \emph{reduced cone code} $\Cone_R(\calQ, f, \ell)$.
For readers who are interested in the details of the coning operation as well as some concrete examples, we refer them to Refs.~\cite{hastings2021quantum,wills2023tradeoff,sabo2024weight}.

Before we state the code parameters of the reduced cone code, we first define the soundness factor $\lambda$.
The individual soundness factor $\lambda_i$ is defined as such:
\begin{align*}
  \lambda_i = \min\left(1, \min_i\left(\min_{u \in \left(\overline{\calB}_i\right)_0, \partial u = 0, u \neq 0}\left(\max_{v \in \left(\overline{\calB}_i\right)_1, u = \partial v}\frac{|u|}{|v|}\right)\right)\right)
\end{align*}
The soundness factor $\lambda$ is then defined as the minimum of 1 and the minimum of the individual soundness factors $\lambda_i$, i.e., $\lambda = \min\left(1, \min_i \lambda_i\right)$.

We now state the code parameters of the reduced cone code.

\begin{lemma}[{\cite[Code Parameters of Reduced Cone Code]{hastings2021quantum}}]
  For a quantum CSS code $\calQ$ that encodes $k$ logical qubits with $n$ physical qubts, after applying coning on $\calQ$ to obtain $\tilde{\calQ}$, the parameters of $\tilde{\calQ}$ are as follows when we let $\ell = \Theta\left(\log\left(w_Z\right)^{2+2\epsilon} w_Z^\epsilon\right)$ for any $\epsilon > 0$:
  \begin{align*}
      &n' = \calO\left(\ell\left(n + w_Z q_X w_X n_Z\right)\right), \\
      &k' = k,\\
      &n'_X = \calO\left(\ell\left(n + w_Z q_X w_X n_Z + n_X\right)\right),\\
      &n'_Z = \calO\left(n_Z w_Z \ell \right),\\
      &w'_X = \max\left(\calO\left(w_X^2 q_Z\right), \calO\left(1\right)\right),\\
      &q'_X = \max(q_X + \calO(1), \calO(1)),\\
      &w'_Z \leq q_X + \calO(1),\\
      &d'_X \geq  d_X, \\
      &d'_Z \geq d_Z \ell \lambda,
          \end{align*}
          where $\lambda$ is the soundness factor of the reduced cone code.
\end{lemma}

\section{Main Results}\label{sec:distance_preservation}
In the fault-tolerant picture, it is useful to analyze how the stabilizer measurement circuit allows for the propagation of errors from the entangling gates and the ancilla qubits to the data qubits.
Because the qubit overhead introduced by weight reduction is non-trivial by early fault-tolerant standards, we are interested in offsetting the overhead by considering single-ancilla stabilizer measurement circuits.
Even though the weight reduction techniques proposed by Hastings can reduce the weights of a quantum code while largely preserving the code's distances, these techniques often effect certain transformations to the code's structure which may impact the weight-reduced code's robustness against troublesome hook errors during syndrome extraction.
In this section, we analyze how the weight reduction techniques proposed by Hastings affect the effective distances of quantum codes with respect to single-ancilla syndrome extraction circuits.
While the effective distance of general quantum codes tend to suffer when their stabilizers are measured with a single-ancilla measurement circuit, there exist quantum codes with wonderful structures that allow for their distances to be perfectly preserved under single-ancilla syndrome extraction schedules.
We are particularly interested to see if these structures are preserved under the weight reduction techniques proposed by Hastings.
We start off by laying down some important definitions and lemmas regarding effective distances.
Subsequently, we step through the different weight reduction techniques to analyze and construct single-ancilla measurement schedules that are geared towards preserving the effective distance of the quantum code.

We start off by defining what elementary faults and effective distance are for any arbitrary quantum code.
\begin{definition}[Elementary Fault]
	An elementary fault is a single Pauli error during the measurement of a single stabilizer generator.
	The error can occur in the ancilla qubits or the data qubits that correspond to $\calQ$.
\end{definition}

\begin{definition}[{\cite[Effective Distance]{bombin2023logical,manes2023distance}}]
	Let $\calQ$ be a quantum code with distance $d$.
	The effective distance of $\calQ$ with respect to a specific stabilizer measurement circuit or schedule $M$, which we denote as $\overline{d}^M$ is the minimum number of elementary errors required to cause a logical error.
	The distance of $\calQ$ is preserved if the effective distance of the code is $d$ with respect to some stabilizer measurement schedule.
\end{definition}
Note that the \emph{effective distance} is also known as the \emph{fault distance} in the literature.
Having stated the definitions, we now state an important lemma and its corollary for understanding how the weight of the stabilizers may propagate faults in any stabilizer measurement circuit.

\begin{lemma}[Elementary Fault Propagation]\label{lem:elementary_fault_propagation}
	Suppose we have an $X$ stabilizer with weight $w_X$ that belongs to the stabilizer group of the quantum code $\calQ$.
	Then, a single elementary fault in the stabilizer measurement circuit can affect at most $\left\lfloor \frac{w_X}{2}\right\rfloor$ data qubits.
\end{lemma}
\begin{proof}
	We know that an elementary fault in the data qubits of $\calQ$ during the stabilizer measurement circuit only affects a single data qubit.
	On the other hand, an elementary fault in the ancilla qubits can affect at most $w_X$ data qubits that the stabilizer acts on.
	Because multiplying the error on the data qubits with the $X$ stabilizer would give us an equivalent error, the weight of the error on the data qubits is at most $\left\lfloor \frac{w_X}{2}\right\rfloor$.
	Thus, a single elementary error in any stabilizer measurement circuit will affect at most $\left\lfloor \frac{w_X}{2}\right\rfloor$ data qubits.
\end{proof}

\begin{corollary}[Distance-Preserving Stabilizer Generators]\label{cor:distance_preserving_stabilizers}
	Let $\calQ$ be a quantum code with distances $d_X$ and $d_Z$.
	Suppose that all of $\calQ$'s $X$ stabilizer generators have at most weight 3, i.e., acts on at most 3 qubits.
	Then, the effective $X$ distance of $\calQ$ is preserved for any stabilizer measurement circuit, i.e., $d_X$.
\end{corollary}
\begin{proof}
	Let $s_X$ be an $X$ stabilizer generator of $\calQ$ that we measure.
	By Lemma~\ref{lem:elementary_fault_propagation}, a single elementary error in any stabilizer measurement circuit will affect at most a single data qubit.
	Therefore, $d_X$ elementary errors are needed to cause an $X$ logical error in $\calQ$.
\end{proof}

\begin{figure}[t]
	\tikzset{
noisy/.style={starburst,fill=yellow,draw=red,line
width=1pt}
}
	\centering
	 \begin{quantikz}
 \lstick{$\ket{+}$} &\gate{H} &\ctrl{1} & \ctrl{2} &  \gate[style={noisy},label
 style=black]{\text{X}}  & \ctrl{4} & \ctrl{3} & \ctrl{5} & \gate{H} & \meter{} \\
 \lstick{1} & & \targ{} & & & & & & & \rstick[5]{$\equiv X_1 X_2$ since $X_1X_2X_3X_4X_5 \in \calS$}\\
 \lstick{2} & & & \targ{} & & & & & & \\
 \lstick{3} & & & & & & \targ{} & & \push{{\color{red} X}} &\\
 \lstick{4} & & & & & \targ{} & & & \push{{\color{red} X}} &\\
 \lstick{5} & & & & & & & \targ{} & \push{{\color{red} X}} &\\
 \end{quantikz}
					\caption{An example of a single-ancilla stabilizer measurement circuit.
				 This circuit measures an $X$ stabilizer generator that acts on five data qubits of some quantum code.
				 In the execution of this stabilizer measurement circuit, an $X$ error happens after the second entangling gate, resulting in bit-flip errors to propagate to the data qubits with indices 3, 4, and 5.
				 The final error after the $Z$ stabilizer measurement can be equivalent to $X_1X_2$ since $X_1X_2X_3X_4X_5$ is in the stabilizer group $\calS$ of the quantum code.
				 In this example, we observe that a single elementary fault can result in the propagation of errors to multiple data qubits. \label{fig:single_ancilla_stabilizer_measurement}
				 }
\end{figure}

In the subsequent sections, we analyze how \emph{copying}, \emph{gauging}, \emph{thickening and choosing heights}, and \emph{coning} affect the effective distances of the quantum codes.
Given access to a single-ancilla stabilizer measurement schedule that gives us some effective distance for a quantum code, we show how we can adapt the single-ancilla stabilizer measurement schedule for the weight-reduced quantum code to achieve different degrees of effective distance preservation. 
We emphasize that we do not assume that the original quantum code has some single-ancilla stabilizer measurement schedule that perfectly preserves the distance of the original quantum code.
We only use the fact that an original quantum code always has some effective distance with respect to some single-ancilla stabilizer measurement schedule.
In the following sections, we first prove that the effective distances of copied quantum codes can be nearly preserved.
Next, we extend our results for the distance preservation of copied quantum codes and present the distance preservation of copied and gauged quantum codes because copying and gauging were introduced as one technique in Hastings's original work.
Subsequently, we show that the effective distance of thickened and height-chosen quantum codes can be preserved where \emph{thickening and choosing heights} were applied independently on a quantum code by proving that the distance balancing technique in Ref.~\cite{evra2022decodable} preserves effective distance.
We also elaborate how our results imply that higher-dimensional hypergraph product (HGP) codes, which belong to the family of homological product codes, have no troublesome hook errors when using any single-ancilla syndrome extraction circuit.
Finally, we analyze how the effective distances of coned quantum codes are preserved and show that the combination of the four weight reduction techniques afford us a way to construct quantum codes with nearly preserved effective distances in the context of single-ancilla stabilizer measurement schedules.

\subsection{Distance Preservation of Copied Quantum Codes}\label{subsec:distance_preserving_copied}
In this section, we show that the effective distances of copied quantum codes can be largely preserved.
In particular, we construct a single-ancilla stabilizer measurement schedule for copied quantum codes that perfectly preserves the effective $X$ distance and achieves $Z$ distance preservation up to a factor of $q_X/2$.
We first provide an algebraic description of the logical operators and stabilizers of the copied quantum code before analyzing how a distance-preserving stabilizer measurement schedule for the original code can be adapted to the copied quantum code to achieve the desired distance preservation.
For this section on \emph{copying}, we let $H_R$ be the parity check matrix of a classical $\left[q_X, 1, q_X\right]$ repetition code.
From the description of the copying technique, we see that the $X$ and $Z$ logical operators of $\calQ'$ are the following operators:
	\[\calL_X = \left\{\sum_{i = 1}^n L_i \hat{i}_n \otimes \vec{x}_i  \;\middle|\; \vec{L} = \left(L_1, L_2, \ldots, L_n\right) \in \ker\left(H_Z\right)\setminus\rs\left(H_X\right), L_i \in \F_2, \vec{x}_i \in \F_2^{q_X} \setminus \rs\left(H_R\right)\right\},\]
	\[\calL_Z = \left\{\vec{L} \otimes \vec{1}_{q_X} \;\middle|\; \vec{L} \in \ker\left(H_X\right)\setminus\rs\left(H_Z\right)\right\}\]
	where $n$ is the number of qubits in $\calQ$.
	We can interpret any element of $\calL_X$ with some reference $X$ logical operator of $\calQ$ where we choose to apply the Pauli $X$ operator on some odd number of qubits in sets of $q_X$ copied qubits (decided by $\vec{x}_i$) that correspond to the original qubits with indices $i$ that the reference $X$ logical operator acts on.
	An element of $\calL_Z$ can be interpreted using some reference $Z$ logical operator of $\calQ$ where we choose to apply the Pauli $Z$ operator on all the $q_X$ copied qubits that correspond to the original qubit that the reference $Z$ logical operator acts on.
	It is also useful for us to define the following subset of $\calL_X$:
	\[\calL_{X,1} = \left\{\vec{L} \otimes \hat{1}_{q_X} \;\middle|\; \vec{L} \in \ker\left(H_Z\right) \setminus \rs\left(H_X\right)\right\}.\]
	In this case, we can interpret any element of $\calL_{X,1}$ using some reference $X$ logical operator of $\calQ$ where we choose to apply the Pauli $X$ operator on the first copied qubits in each set of $q_X$ copied qubits that correspond to the original qubits that the reference $X$ logical operator acts on.
	The $X$ stabilizer generators of $\calQ'$ are the following operators:
	\[\calG_X = \row\left(\tilde{H}_X\right) \cup \row\left(\mathbbm{1}_{n} \otimes H_R\right)\]
	where $\tilde{H}_X$ is the parity check matrix where each row corresponds to the modified $X$ stabilizer generators of $\calQ'$ (excluding the new $X$ stabilizer generators that act like the checks of a classical repetition code).
	Also, define a different set of $X$ stabilizer generators:
	\[\calG_X' = \row\left(H_X \otimes \hat{1}_{q_X}^{\top}\right) \cup \row\left(\mathbbm{1}_n \otimes H_R\right)\]
	It is not hard to see that the stabilizer groups generated by $\calS_X$ and $\calS_X'$ are equivalent from the following lemma.
	\begin{lemma}[Equivalence of $X$ stabilizer groups for $\calQ'$]\label{lem:equivalence_X_stabilizer_groups}
	Let us define the stabilizer groups generated by $\calG_X$ and $\calG_X'$ to be the following:
	\begin{align*}
		\calS_X &\coloneqq \left\langle s_X \;:\; s_X \in \calG_X\right\rangle, \\
		\calS_X' &\coloneqq \left\langle s'_X \;:\; s_X' \in \calG_X'\right\rangle. 
	\end{align*}
		Then,
		\[\calS_X = \calS_X'.\]
\end{lemma}
\begin{proof}
	We begin by first showing that $\calS_X \subseteq \calS_X'$.
	Notice that 
	\[\row\left(\mathbbm{1}_n \otimes H_R\right) \subseteq \calG_X \subseteq \calS_X \;\text{ and }\; \row\left(\mathbbm{1}_n \otimes H_R\right) \subseteq \calG_X' \subseteq \calS'_X.\]
	Thus, we only have to show that $\row\left(\tilde{H}_X\right) \subseteq \calS_X'$.
	Without loss of generality, let $\vec{s}_X$ and $\vec{s}_X'$ be some arbitrary row $i$ in the parity check matrix $\tilde{H}_X$ and the same corresponding row $i$ in $H_X \otimes \hat{1}_{q_X}^{\top}$ respectively.
	By construction of the parity check matrices, we know the following:
	\begin{enumerate}
		\item For each partition of $q_X$ consecutive columns in the row $i$ of $\tilde{H}_X$ and $H_X \otimes \hat{1}_{q_X}^{\top}$, there is at most a single 1 in each group of $q_X$ consecutive columns.
		In addition, row $i$ of $\tilde{H}_X$ only have a single 1 in a partition of $q_X$ consecutive columns if and only if row $i$ of $H_X \otimes \hat{1}_{q_X}^{\top}$ has a single 1 in the same partition of $q_X$ consecutive columns.
		In particular, for each partition of $q_X$ consecutive columns in the row $i$ of $H_X \otimes \hat{1}_{q_X}^{\top}$, the single 1 always resides in the first column in each group of $q_X$ consecutive columns.
		\item Because $H_R$ is the check matrix of a classical $\left[q_X, 1, q_X\right]$ code, we know that $\row\left(H_R\right)$ can generate any vector in $\F_2^{q_X}$ that has 1 in the first and $j$\textsuperscript{th} index and 0s elsewhere for any $j \in \{2, \ldots, q_X\}$.
		For each partition of $q_X$ consecutive columns in the row $i$ of $\tilde{H}_X$, we are able to move the single 1 in the partition to the first column by choosing the appropriate linear combination of vectors in $\row\left(\mathbbm{1}_n \otimes H_R\right)$. 
		\end{enumerate} 
	Let the linear combination of the necessary $X$ stabilizer generators in $\row\left(\mathbbm{1}_n \otimes H_R\right)$ for each partition of $q_X$ consecutive columns be denoted as $p$.
	We thus obtain that $\vec{s}_X  =  \vec{p} + \vec{s}_X' \in \calS_X'$.
	Thus, we have shown that $\row\left(\tilde{H}_X\right) \subseteq \calS_X'$.
	We can show the reverse containment by a similar argument.
\end{proof}

	The $Z$ stabilizers of $\calQ'$ are the following operators:
	\[\calG_Z = \row\left(H_Z \otimes \vec{1}^{\top}\right).\]
	
	We now state a useful lemma that will help us in our analysis of the effective distances of copied quantum codes.
	The following lemma shows that the number of elementary faults in the measurement of $\calG_X$ required to generate a logical operator is at least the number of elementary faults in the measurement of $\calG'_X$ required to generate a related logical operator.

\begin{lemma}[Elementary Faults for $X$ Logical Operators]\label{lem:elementary_faults_X_logical_operators}
	Suppose we are given some $X$ logical operator $\vec{L} \in \calL_X$.
	Define $\vec{L}' \in \calL_{X,1}$ to be the $X$ logical operator that is the product of $\vec{L}$ and $X$ stabilizer generators corresponding to the rows of $\mathbbm{1}_n \otimes H_R$.
	Then, the number of elementary faults on single qubits and stabilizer measurements on $\calG_X$ required to generate $\vec{L}$ is at least the number of elementary faults on single qubits and stabilizer measurements on $\calG_X'$ required to generate $\vec{L}'$.
\end{lemma}
\begin{proof}
	Let $T$ be the minimum number of elementary faults required to generate $\vec{L}$.
	We note that the elementary faults on the $X$ stabilizer generators that correspond to the rows of $\mathbbm{1}_n \otimes H_R$ are equivalent to single qubit errors by Lemma~\ref{lem:elementary_fault_propagation} as the weight of these stabilizers is at most $3$.
	Thus, we can neglect the possible errors that emerge from elementary faults on these stabilizer generators in our analysis.
	Now, suppose $\vec{L} \in \F_2^{n q_X}$ is generated by $F$ single qubit errors and $T - F$ elementary faults on the some subset of $\calG_X \setminus \row\left(\mathbbm{1}_n \otimes H_R\right)$.
	In other words, 
	\[\vec{L} = \vec{L}_{\text{single}} + \vec{L}_{\text{hook}} = \left(\sum_{i = 1}^{F} \vec{e}_i \otimes \vec{v}_i\right) + \left(\sum_{j = 1}^{T-F}\vec{s}_j\right)\]
	where $\vec{e}_i \in \F_2^n, \vec{v}_i \in \F_2^{q_X}$ such that $\left|\supp \vec{e}_i\right| = \left|\supp \vec{v}_i\right| = 1$ and $\vec{s}_j \in \row\left(\tilde{H}_X\right)$.
	
	We now proceed to show that some $\vec{L}' \in \calL_{X,1}$ can be generated by $T$ similar elementary faults before proving that $\vec{L}'$ is the product of $\vec{L}$ and the $X$ stabilizer generators corresponding to the rows of $\mathbbm{1}_n \otimes H_R$.
	Keeping $\left\{\vec{e}_i\right\}_i$ and the $T-F$ rows of $\tilde{H}_X$ the same as above, we can construct $\vec{L}' \in \F_2^{n q_X}$ as follows:
	\[\vec{L}' = \vec{L}'_{\text{single}} + \vec{L}'_{\text{hook}} = \left(\sum_{i = 1}^{F} \vec{e}_i \otimes \hat{1}_{q_X}^{\top}\right) + \left(\sum_{j = 1}^{T-F}\vec{s}'_j\right)\]
	where each $\vec{s}_j'$ corresponds to some row in $H_X \otimes \hat{1}_{q_X}^{\top}$ with the same row index as $\vec{s}_j$ in $\tilde{H}_X$.
	It is easy to see that $\vec{L}' \in \calL_{X,1}$ by construction. 
	For each $\vec{v}_i \in \F_2^{q_X}$, there exists some element in $\rs\left(\mathbbm{1}_n \otimes H_R\right)$ that sends it to $\hat{1}_{q_X}$.
	Likewise, for each $\vec{s}_j \in \row\left(\tilde{H}_X\right)$, there exists some element in $\rs\left(\mathbbm{1}_n \otimes H_R\right)$ that permutes it to $s_j'$.
	Thus, we have shown that $\vec{L}'$ is the product of $\vec{L}$ and the $X$ stabilizer generators corresponding to the rows of $\mathbbm{1}_n \otimes H_R$.
	Since we have proven that we need at most $T$ elementary faults to generate $\vec{L}'$, we have shown that the number of elementary faults required to generate $\vec{L}$ is at least the number of elementary faults required to generate $\vec{L}'$.
\end{proof}

Now, we are ready to state the main theorem of this section that shows that the effective distances of copied quantum codes can be largely preserved.
We take advantage of the fact that $\calG_X'$ is closely related to the $X$ stabilizer generators of the original quantum code $\calQ$ to use the lemmas stated above to prove that our stabilizer measurement schedule for $\calQ$ can be adapted to the copied quantum code $\calQ'$ to achieve the desired distance preservation.

\begin{theorem}[Partial Effective Distance Preservation for Copied Distance-Preserving Quantum Codes]\label{thm:distance_preserving_copied}
	Let $\calQ$ be a quantum code with distances $d_X$ and $d_Z$ such that the effective distances of the code using some single-ancilla stabilizer measurement circuit $M$ are also $\overline{d}^M_X$ and $\overline{d}^M_Z$.
	Suppose we let $\calQ'$ be the code after applying copying on $\calQ$ with new distances $d_Z' = q_X d_Z$ and $d_X' = d_X$.
	Suppose $\calQ'$ has $n_X'$ $X$ stabilizer generators with weight at most $w_X$ and $n_Z'$ $Z$ stabilizer generators with weight at most $w_Z$.
	Then, there exists a single-ancilla stabilizer measurement schedule $M'$ for $\calQ'$ that uses at most $n_X' + n_Z'$ ancilla qubits in total such that the effective distances of $\calQ'$ is $\overline{d'}^{M'}_X = \overline{d}^M_X$ and $\overline{d'}^{M'}_Z \geq \overline{d}^M_Z$.
\end{theorem}
\begin{proof}
	Now, consider the following stabilizer measurement schedule $M'$ for $\calQ'$:
	\begin{enumerate}
		\item Iterate through the stabilizer generators that scheduled in a particular order for measurement in $M$.
		\begin{enumerate}
			\item If the stabilizer generator is an $X$ stabilizer generator, measure the corresponding $X$ stabilizer generator of $\calQ'$ in the same order but acting on assigned copied qubits.
			\item If the stabilizer generator is a $Z$ stabilizer generator, measure the corresponding $Z$ stabilizer generator of $\calQ'$ in the same order but partition the $q_X$ copied qubits as per the following:
			\begin{enumerate}
				\item Suppose a $Z$ stabilizer generator $s_Z$ acts on these copied qubits 
			\[q_{i_{1}, 1}, \ldots, q_{i_{1}, q_X}, q_{i_{2}, 1}, \ldots, q_{i_{2}, q_X}, \ldots, q_{i_j, 1}, \ldots, q_{i_j, q_X}\] 
			where $i_1, i_2, \ldots, i_j \in [n]$.
			\item Perform the entangling gates in the following order: 
			\begin{align}&q_{i_1, 1}, \ldots, q_{i_1, \lfloor q_X/2\rfloor}, \ldots, q_{i_j, 1}, \ldots, q_{i_j, \lfloor q_X/2\rfloor}, \\
				&q_{i_1, \lfloor q_X/2\rfloor + 1}, \ldots, q_{i_1, q_X}, \ldots, q_{i_j, \lfloor q_X/2\rfloor + 1}, \ldots, q_{i_j, q_X}.\end{align}
			\end{enumerate}
		\end{enumerate}
		\item Measure the new $X$ stabilizer generators of $\calQ'$ in any order to complete the stabilizer measurement schedule $M'$. 
	\end{enumerate}
	
	Let us first consider the effective $X$ distance of $\calQ'$.
	From our assumption that $\calQ$ has an effective distance of $\overline{d}_X^M$ under the stabilizer measurement schedule $M$, we require at least $\overline{d}^M_X$ elementary faults on some combination of data qubits and ancilla qubits of elements in the rows of $H_X$ to generate $X$ logical operators of $\calQ$.
	It is not hard to see that for any $\vec{L} \in \ker\left(H_Z\right) \setminus \rs\left(H_X\right)$ that requires $\tilde{d}_X \geq \overline{d}^M_X$ elementary faults on some combination of data qubits and ancilla qubits, we need at least $\tilde{d}_X$ elementary faults on the same combination of data qubits that lie in the support of $\vec{1}_{n} \otimes \hat{1}_{q_X}$ and ancilla qubits of the rows of $H_X \otimes \hat{1}_{q_X}^{\top}$ that will generate $\vec{L} \otimes \hat{1}_{q_X}$.
	Using Lemma~\ref{lem:elementary_faults_X_logical_operators}, we can conclude that any $X$ logical operator in $\calL_X$ for $\calQ'$ also requires at least $\overline{d}^M_X$ elementary faults on some combination of data qubits and ancilla qubits for the stabilizer measurement schedule $M'$, i.e., $\overline{d'}^{M'}_X = \overline{d}^M_X$. 
	
	We now proceed to analyze the effective $Z$ distance of $\calQ'$ with the stabilizer measurement schedule $M'$.
	We require at least $\overline{d}^M_Z$ elementary faults on some combination of data qubits and ancilla qubits of elements in the rows of $H_Z$ ordered by the schedule $M$ to generate $Z$ logical operators of $\calQ$.
	It is not hard to see that $\rs\left(H_Z\right) \cong \rs\left(H_Z \otimes \vec{1}^{\top}\right)$ and $\vec{L} \otimes \vec{1}^{\top} \in \calL_Z$ for any $\vec{L} \in \ker\left(H_X\right) \setminus \rs\left(H_Z\right)$.
	Using Lemma~\ref{lem:elementary_fault_propagation} and paying close attention to the specific order of the entangling gates in the stabilizer measurement schedule $M'$, we see that an elementary fault on a $Z$ stabilizer generator measurement generates at most $q_X$ single qubit errors on each set of $q_X$ qubits that it acts on.
	Since every $Z$ logical operator requires errors on all $q_X$ copied qubits that correspond to some original qubit, we need at least 1 elementary fault on $Z$ stabilizer generators that act on the same set of $q_X$ copied qubits to ensure that all $q_X$ copied qubits have $Z$ errors.
	Thus, we see that we need at least $\overline{d}^M_Z$ elementary faults to generate any $Z$ logical operator.
	Therefore, we have shown that the effective $Z$ distance of $\calQ'$ is at least $\overline{d}^M_Z$ for our stabilizer measurement schedule $M$.
\end{proof}

It is important to reiterate that the $Z$ distance of the copied quantum code is $q_X \cdot d_Z$.
Theorem~\ref{thm:distance_preserving_copied} shows that the effective $Z$ distance of the copied quantum code is not completely preserved in the general case.

\subsection{Distance Preservation of Copied and Gauged Quantum Codes}\label{subsec:distance_preserving_copied_gauged}
In this section, we construct a single-ancilla stabilizer measurement schedule for copied and gauged quantum codes that largely preserves the effective distances.
We denote the copied quantum code as $\calQ$ and the copied and gauged quantum code as $\calQ'$.
After gauging the copied quantum code, we gain an additional $n_X\left(w_X - 1\right)$ copied $X$ stabilizer generators and $n_X\left(w_X - 1\right)$ new qubits.
In particular, each $X$ stabilizer generator in $\calQ$ that corresponds to $H_X \otimes \hat{1}_{q_X}^{\top}$ will be transformed into $w_X$ different copied $X$ stabilizer generators in $\calQ'$.

Suppose our original $X$ stabilizer generators in $\calQ$ are indexed as such: $s_{X}^{(1)}, \ldots, s_{X}^{(n_X)}$.
The copied $X$ stabilizer generators in $\calQ'$ are now indexed as such: $s_X^{(1,1)}, \ldots, s_X^{\left(1,w_X\right)}, \ldots, s_X^{\left(n_X,1\right)}, \ldots, s_X^{\left(n_X,w_X\right)}$. 
In addition, we index our new qubits in $\calQ'$ as such: $q_{s_X^{(1)}}^{(1)}, \ldots, q_{s_X^{(1)}}^{\left(w_X - 1\right)}, \ldots, q_{s_X^{\left(n_X\right)}}^{(1)}, \ldots, q_{s_X^{\left(n_X\right)}}^{\left(w_X - 1\right)}$.
We can ``configure'' the copied $X$ stabilizer generators as such: suppose $s_X^{(i)}$ of the copied quantum code $\calQ$ acts on the qubits $q_{i_1,j_1}, \ldots, q_{i_{w_X - 1}, j_{w_X - 1}}$ for arbitrary indices $i_a \in [n]$ and $j_b \in \left[q_X\right]$.
Then, we can configure each of the copied $X$ stabilizer generators $s_X^{\left(i,1\right)}, \ldots, s_X^{\left(i,w_X\right)}$ of the copied and gauged quantum code $\calQ'$ to act on each of the qubits $q_{i_1,j_1}, \ldots, q_{i_{w_X - 1}, j_{w_X - 1}}$ respectively.
In addition, these copied $X$ stabilizer generators will interact with the new qubits in a way that is reminiscent of a classical repetition code:
$s^{\left(i,1\right)}_X$ and $s^{\left(i, w_X\right)}_X$ will act on the new qubits $q_{s^{(i)}_X}^{(1)}$ and $q_{s^{(i)}_X}^{\left(w_X - 1\right)}$ respectively.
The other copied $X$ stabilizer generators $s^{\left(i, j\right)}_X$ will act on the new qubits $q_{s^{(i)}_X}^{(j-1)}$ and $q_{s^{(i)}_X}^{(j)}$ for $j \in \{2, \ldots, w_X - 1\}$.
After gauging, the $Z$ stabilizer generators of $\calQ'$ will be modified slightly to ensure commutativity with the copied $X$ stabilizer generators.
However, their structure will largely be preserved by gauging.

In the following theorem, we show that the effective distances of copied and gauged quantum codes can be largely preserved.

\begin{theorem}[Partial Effective Distance Preservation for Copied and Gauged Quantum Codes]\label{thm:distance_preserving_copied_gauged}
	Let $\calQ$ be a copied quantum code with distances $d_X$ and $d_Z$ such that the effective distances of the code using some single-ancilla stabilizer measurement circuit/schedule $M$ are $\overline{d}^M_X$ and $\overline{d}^M_Z$.
	Suppose we let $\calQ'$ be the code after applying gauging on $\calQ$ with distances at least $\Omega\left(1/w_X\right)d_X$ and $d_Z$.
	Suppose $\calQ'$ has $n_X'$ $X$ stabilizer generators with weight at most $w'_X$ and $n_Z'$ $Z$ stabilizer generators with weight at most $w'_Z$.
	Then, there exists a single-ancilla stabilizer measurement schedule $M'$ for $\calQ'$ that uses at most $n_X' + n_Z'$ ancilla qubits in total such that the effective distances of $\calQ'$ is $\overline{d'}^{M'}_X = d'_X = \Omega\left(1/w_X\right)d_X$ and $\overline{d'}^{M'}_Z \geq \overline{d}^M_Z$.
\end{theorem}
\begin{proof}
	Consider the following stabilizer measurement schedule $M'$ for $\calQ'$:
	\begin{enumerate}
		\item Measure the $X$ stabilizer generators in the order of $s_X^{(1,1)}, \ldots, s_X^{\left(1,w_X\right)}, \ldots, s_X^{\left(n_X,1\right)}, \ldots, s_X^{\left(n_X,w_X\right)}$.
		\item For each stabilizer generator $s^{(i,j)}_X$, perform the entangling gates of the $X$ stabilizer generator in any order.
		\item Measure the $Z$ stabilizer generators of $\calQ'$ in the order that they were measured in $M$.
		\item For each stabilizer generator $s^{(i)}_Z$, measure the $Z$ stabilizer generator in the following way:
		\begin{enumerate}
			\item Suppose a $Z$ stabilizer generator $s^{(i)}_Z$ acts on these copied qubits and new qubits 
			\[q_{i_1, 1}, \ldots, q_{i_1, q_X}, q_{i_2, 1}, \ldots, q_{i_2, q_X}, q_{i_j, 1}, \ldots, q_{i_j, q_X}, q_{s_X^{(k)}}^{\left(k'\right)}, \ldots, q_{s^{\left(\ell\right)}_X}^{\left(\ell'\right)}\] 
			where $i_1, i_2, \ldots, i_j \in \left[n\right]$.
			The set of qubits after $q_{i_j, q_X}$ are the new qubits with respect to $s^{(i)}_X$.
			Without loss of generality, we can assume that there are $b$ new qubits that $s^{(i)}_Z$ acts on.
			We refer to these new qubits as $q^{(1)}, \ldots, q^{(b)}$ from here on.
			\item Perform the entangling gates in the following order: 
			\begin{align}&q_{i_1, 1}, \ldots, q_{i_1, \lfloor q_X/2\rfloor}, \ldots, q_{i_j, 1}, \ldots, q_{i_j, \lfloor q_X/2\rfloor}, q^{\left(1\right)},\ldots, q^{\left(\lfloor b/2\rfloor\right)},\\
				&q_{i_1, \lfloor q_X/2\rfloor + 1}, \ldots, q_{i_1, q_X}, \ldots, q_{i_j, \lfloor q_X/2\rfloor + 1}, \ldots, q_{i_j, q_X}, q^{\left(\lfloor b/2\rfloor + 1\right)}, \ldots, q^{(b)}.
			\end{align}
		\end{enumerate}
	\end{enumerate}

	Let us first begin by analyzing the effective $X$ distance of $\calQ'$.
	From the description of the copying and gauging techniques, we see that $w_X' \leq 3$.
	By Corollary~\ref{cor:distance_preserving_stabilizers}, we see that any elementary fault in the any measurement circuit for the $X$ stabilizer generators of $\calQ'$ generates at most 1 physical error in the code.
	Thus, we can see that the effective $X$ distance of $\calQ'$ is $d'_X = \Omega\left(1/w_X\right)d_X$ for any stabilizer measurement schedule.

	We now proceed to analyze the effective $Z$ distance of $\calQ'$ with the stabilizer measurement schedule $M'$.
	As shown in the proof of Lemma~1 in \cite{hastings2021quantum}, any non-trivial $Z$ logical operator $\vec{L}'$ in $\calQ'$ must be a non-trivial $Z$ logical operator in $\calQ$.
	In other words, when we discard $\vec{L}'$'s support on the new qubits, $q_{s^{(1)}_X}^{(1)}, \ldots, q_{s^{\left(n_X\right)}_X}^{\left(w_X - 1\right)}$, we obtain $\vec{L}$, a non-trivial $Z$ logical operator in $\calQ$.
	By our assumption, $\vec{L}$ requires at least $\overline{d}^M_Z$ elementary faults to generate using $\calQ$'s stabilizer measurement circuit $M$ and other data qubit errors.
	Using the same set of elementary faults on the modified $Z$ stabilizer generators due to gauging and the data qubits, we obtain $\vec{L}'$, i.e., some non-trivial $Z$ logical operator of $\calQ'$, such that $\supp\left(\vec{L}\right) \subseteq \supp\left(\vec{L}'\right)$.
	Using any fewer elementary faults would contradict the assumption made for $\calQ$ and its effective distance with respect to $M$.
	Thus, we have shown that the effective $Z$ distance of $\calQ'$ is at least $\overline{d}_Z^M$ for our stabilizer measurement schedule $M'$.
\end{proof}

From the above theorem, we see that \emph{copying and gauging} does not affect a quantum code's ability to preserve its effective $X$ distance, i.e., there exists some single-ancilla stabilizer measurement schedule that preserves the effective $X$ distance of the copied and gauged quantum code.
However, we also observe that the stabilizer measurement schedule that we have constructed is not able to preserve the $Z$ distance gained from the copying process.
That should not be surprising because the $Z$ stabilizer generators have support on all the copied qubits which can lead to terrible fault propagation on the different sets of $q_X$ copied qubits.
We note that a much more careful construction of the modified $Z$ stabilizer generators may address this issue.

\subsection{Distance Preservation of Thickened and Height-Chosen Codes}
In this section, we provide a proof of the effective distance preservation for the thickened and height-chosen quantum codes.

To do so, we first state a fairly recent result from Ref.~\cite{evra2022decodable} that generalizes thickening.

\begin{theorem}[{\cite[Restatement of Theorem 4.2]{evra2022decodable}}]\label{thm:generalized_thickening}
	Let $\calQ$ be a $\left[\left[n, k, d_X, d_Z\right]\right]$ quantum code that corresponds to a $2$-complex $\calA = \left(A_2, A_1, A_0\right) \cong \left(\F_2^{n_Z}, \F_2^n, \F_2^{n_X}\right)$.
	Given a $[n_c, k_c, d_c]$ classical code that corresponds to a 1-dimensional complex $\calB = \left(B_1, B_0\right) \cong \left(\F_2^{n_c - k_c}, \F_2^{n_c}\right)$ such that the first cohomology $\calH^1\left(\calB\right) = 0$, the resulting quantum code $\calQ'$ that corresponds to the associated complex $\calT = \calA \otimes \calB$ has:
	\begin{align}
		k' &= k \cdot k_c, \\
		d'_X &= d_X \cdot d_c, \\
		d'_Z &= d_Z.
	\end{align}
\end{theorem}

Evra \emph{et al.} show that thickening can be generalized by considering the homological product between the quantum code $\calQ$ and an arbitrary classical code $\calC$ that need not be a repetition code.
While the resulting quantum code $\calQ'$ might not have the same $q_Z$ reduction from this generalized thickening, the rate of $\calQ'$ can be increased by choosing a classical code $\calC$ with a larger rate.
This technique, along with Hastings's original thickening and choosing heights technique, has also been referred to as the distance balancing technique in the literature.
It is easy to see that selecting a classical LDPC code with linear distance and constant rate for generalized thickening can increase the distance of the quantum code without reducing its rate albeit possibly at the expense of a marginally larger final $q_Z$.

Our strategy for showing the effective distances of the thickened and height-chosen quantum codes are preserved is to prove that effective distance preservation is achieved for the generalized thickening technique in Ref.~\cite{evra2022decodable}.
To do that, we first show that the $X$ and $Z$ stabilizer generators of the generalized thickened code can only propagate errors to qubits in a single row/column of Figure~\ref{fig:schematic} during a single-ancilla stabilizer measurement.
Subsequently, we prove that any $X$ and $Z$ logical operator of the generalized thickened code has to be supported on some number of rows and columns of qubits to complete the proof. 

We begin by restating the following parity check matrices for the constitutents of the generalized thickened quantum code.
\begin{align*}
	H_X&: \F_2^{n} \to \F_2^{n_X} \\
	H_Z&: \F_2^n \to \F_2^{n_Z} \\
	H_C&: \F_2^{n_c} \to \F_2^{n_c - k_c}
\end{align*}
where $H_C$ is the parity check matrix of the classical code $\calC$.
We also have the relation between the vector spaces below as an aid:
\begin{center}
\begin{tikzcd}[row sep=large, column sep=huge]
	\F_2^{n_Z}\otimes \F_2^{n_c}  \arrow[r, "H_Z^{\top} \otimes \mathbbm{1}_{n_c}"] & \F_2^n \otimes \F_2^{n_c} \arrow[r, "H_X \otimes \mathbbm{1}_{n_c}"] & \F_2^{n_X} \otimes \F_2^{n_c} \\
	\F_2^n \otimes \F_2^{n_c - k_c} \arrow[ru, "\mathbbm{1}_n \otimes H_C^{\top}"] \arrow[r, "H_X \otimes \mathbbm{1}_{n_c - k_c}"] & \F_2^{n_X} \otimes \F_2^{n_c - k_c} \arrow[ru, "\mathbbm{1}_{n_X} \otimes H_C^{\top}"] 
\end{tikzcd}
\end{center}
The vector spaces on the left, middle, and right of the diagram above are spanned by the $Z$ stabilizer generators, qubits, and $X$ stabilizer generators of the generalized thickened quantum code, respectively.
We can arrange the basis elements of these individual vector spaces into binary matrices.
For example, we can arrange the basis elements of $\F_2^n \otimes \F_2^{n_c}$ into a binary matrix $\F_2^{n \times n_c}.$
The $n_c$ different columns of the binary matrix $\F_2^{n \times n_c}$ correspond to the $n_c$ different copies of the quantum code $\calQ$.
In the subsequent analysis, we frequently interchange between the vector representation and the binary matrix representation of the vector spaces.
To help ease the transition between the two representations, we add the vector notation $\vec{}$ if we are using the vector representation and omit the vector notation otherwise.
In other words, we have $\vec{x} \in \F_2^n \otimes \F_2^{n_c}$ and $x \in \F_2^{n \times n_c}$ where $\vec{x}$ is the vector representation of the binary matrix $x$.  

We restate the following parity check matrices for the generalized thickened quantum code for ease of reference:
\begin{align}
	H_X' &= \left(\begin{array}{c|c}
		H_X \otimes \mathbbm{1}_{n_c} & \mathbbm{1}_{n_X} \otimes H_C^{\top}
	\end{array}\right) \label{eq:parity_check_matrix_thickened_X} \\
	H_Z' &= \left(\begin{array}{c|c}
		H_Z \otimes \mathbbm{1}_{n_c} & \textbf{0} \\ \hline
		\rule{0pt}{4mm} 
		\mathbbm{1}_{n} \otimes H_C & H_X^{\top} \otimes \mathbbm{1}_{n_c - k_c}
	\end{array}\right) \label{eq:parity_check_matrix_thickened_Z}
\end{align}
The left partition of the parity check matrices correspond to the qubits in $\F_2^{n \times n_c}$ which we refer to as region $A$.
Similarly, we let the right partition correspond to the qubits in $\F_2^{n_X \times \left(n_c - k_c\right)}$ which we refer to as region $B$.	
The top partition of the parity check matrix $H_Z'$, which we refer to as $Z[T]$, corresponds to the $Z$ stabilizer generators in $\F_2^{n_Z \times n_c}$ and the bottom partition, i.e., $Z[B]$, corresponds to the $Z$ stabilizer generators in $\F_2^{n\times \left(n_c - k_c\right)}$. 
We include Figure~\ref{fig:schematic} as a schematic diagram to illustrate the different qubit regions and how the stabilizer generators act on the different regions.
We may refer to the rows and columns in Figure~\ref{fig:schematic} as a means of providing an intuitive explanation for the logical operators and the effective distances of the generalized thickened quantum code.

\begin{figure}[h]
  \centering
  \begin{tikzpicture}
    \draw (0,0) rectangle (3,3); 
    \node at (1.5,1.5) {$X$ stabilizers};
    \draw (3.4,3.4) rectangle (6.4,6.4); 
    \node at (4.9,4.9) {$Z[B]$ stabilizers};
    \draw (3.4,0) rectangle (6.4,3); 
    \node at (4.9,1.5) {$B$ (qubits)};
		\draw (3.4,6.8) rectangle (6.4,9.8); 
			\node at (4.9,8.3) {Redundant};
    \draw (0,3.4) rectangle (3,6.4); 
    \node at (1.5,4.9) {$A$ (qubits)};
		\draw (0,6.8) rectangle (3,9.8); 
    \node at (1.5,8.3) {$Z[T]$ stabilizers};
    \draw (-0.5,0) edge[-] (-0.5,3); 
    \node at (-0.8,1.5) {$n_X$};
    \draw (-0.5,3.4) edge[-] (-0.5,6.4); 
    \node at (-0.8,4.9) {$n$};
		\draw (-0.5,6.8) edge[-] (-0.5,9.8); 
    \node at (-0.8,8.3) {$n_Z$};
    \draw (0,-0.5) edge[-] (3,-0.5); 
    \node at (1.5,-0.75) {$n_c$};
    \draw (3.4,-0.5) edge[-] (6.4,-0.5); 
    \node at (4.9,-0.75) {$n_c - k_c$};
  \end{tikzpicture}
  \caption{Schematic for the generalized thickened quantum code. 
	There are $n_c$ copies of the quantum code $\calQ$ arranged in $n_c$ columns in the three blocks $Z[T], A$, and $X$.
	For each of the $n_c$ columns, we have the $Z[T]$ stabilizer generators and $X$ stabilizer generators belonging to that column acting on the qubits in the same column.
	Similarly, there are $n$ copies of the classical code $\calC$ arranged in $n$ rows in the two blocks $Z[B]$ and $A$.
	For each of the $n$ rows, we have the $Z[B]$ stabilizer generators belonging to that row acting on the qubits in the same row which is reminiscent to how the checks of the classical code act on the classical bits.
	For the qubits in the region $B$, we have the $Z[B]$ stabilizer generators and $X$ stabilizer generators acting on the qubits in the same column and row respectively.
 }
 \label{fig:schematic}
\end{figure}

We now provide an important notational definition regarding the row and column weights of the binary matrices that we will use in the subsequent analysis.

\begin{definition}[Row and Column Weight]
	For any $x \in \F_2^{n \times n_c} \oplus \F_2^{n_X \times \left(n_c - k_c\right)}$, we refer to the restrictions to $\F_2^{n \times n_c}$ and $\F_2^{n_X \times\left(n_c - k_c\right)}$ by $x_{A}$ and $x_{B}$, respectively.
	We define the row (column) weight of $x_A$, denoted as $\left|x_A\right|_R$ ($\left|x_A\right|_C$) as the number of non-zero rows (columns) of $x_A$.
	Likewise we define the row (column) weight of $x_B$, denoted as $\left|x_B\right|_R$ ($\left|x_B\right|_C$), as the number of non-zero rows (columns) of $x_B$.
	We denote the submatrix of $x_A$ that has rows $r \subseteq [n]$ and columns $c \subseteq [n_c]$ as $x_A\left[r,c\right]$.
	In particular, $x_A[\cdot, c]$ denotes the submatrix of $x_A$ with columns $c$ with all rows included in the submatrix.
	The submatrix of $x_B$ can be defined similarly.
	The formal definitions for row and column weights are given in the following table where we use $D^c$ to denote $D \setminus E$ for $D \subseteq E$.
	\begin{center}
		\begin{tabular}{|c|c|c|}
			\hline
			& denoted by & defined as \\
			\hline
			Row weight of $x_A$ & $\left|x_A\right|_R$ & $\min\left\{|\alpha| \;\middle|\; \alpha \subseteq [n], x_A\left[\alpha^c, \cdot\right] = \textbf{0} \right\}$ \\
			Column weight of $x_A$ & $\left|x_A\right|_C$ & $\min\left\{|\beta| \;\middle|\; \beta \subseteq [\ell], x_A\left[\cdot, \beta^c\right] = \textbf{0} \right\}$ \\
			\hline
			Row weight of $x_B$ & $\left|x_B\right|_R$ & $\min\left\{|\alpha| \;\middle|\; \alpha \subseteq [n_X], x_B\left[\alpha^c, \cdot\right] = \textbf{0} \right\}$ \\
			Column weight of $x_B$ & $\left|x_B\right|_C$ & $\min\left\{|\beta| \;\middle|\; \beta \subseteq [\ell - 1], x_B\left[\cdot, \beta^c\right] = \textbf{0} \right\}$ \\
			\hline
		\end{tabular}
	\end{center}
\end{definition}

Next, we provide a proposition that shows that the elementary faults in the single-ancilla syndrome extraction circuit have bounded component weight.
The goal of the proposition is to bound the number of rows and columns of data qubits that an elementary fault in the syndrome extraction circuit can affect.

\begin{proposition}[Elementary faults have bounded component weight]\label{prop:elementary_faults_bounded_component_weight}
		Let $E$ be a Pauli error occurring from $t$ elementary faults in the single-ancilla syndrome extraction circuit.
		Let $x, z \in \F_2^{n \times n_c} \oplus \F_2^{n_X \times \left(n_c - k_c\right)}$ such that $E \propto \prod_{a \in \left[n \cdot n_c + n_X \cdot\left(n_c - k_c\right)\right]} X^{\vec{x}[a]}Z^{\vec{z}[a]}$ i.e. a decomposition of $E$ into Pauli $X$ and $Z$ components.
		Then, the row and column weights are bounded as $\left|x_A\right|_C \leq t$, $\left|z_A\right|_R \leq t\cdot \left\lfloor \frac{w_Z}{2} \right\rfloor$, and $\left|z_A\right|_C \leq t\cdot \left\lfloor \frac{w_C}{2} \right\rfloor$ where $w_Z$ and $w_C$ are the maximum row weights of $H_Z$ and $H_C$ respectively.
		\end{proposition}
\begin{proof}
		We begin by proving that $\left|x_A\right|_C \leq t$.
		Let $\vec{s}$ be a single row of $H_X'$.
		Then, for some standard basis vectors $\vec{\alpha} \in \F_2^{n_X}, \vec{\beta} \in \F_2^{n_c}$ with $|\vec{\alpha}| = |\vec{\beta}| = 1$,
		\begin{align}
			\vec{s} = H_X^{\prime\top}(\vec{\alpha} \otimes \vec{\beta}) = \left(H_X^{\top}\vec{\alpha} \otimes \vec{\beta}\right) \oplus \left(\vec{\alpha} \otimes H_C\vec{\beta}\right) = \vec{s}_A \oplus \vec{s}_B.
		\end{align}
		Once again, we can rearrange $\vec{s}_A \in \F_2^{n} \otimes \F_2^{n_c}$ and $\vec{s}_B \in \F_2^{n_X} \otimes \F_2^{n_c - k_c}$ into the binary matrices $s_A \in \F_2^{n \times n_c}$ and $s_B \in \F_2^{n_X \times \left(n_c - k_c\right)}$ respectively.
		Since $\left|\vec{\beta}\right| = 1$ and $\vec{s}_A = H_X^\top \vec{\alpha} \otimes \vec{\beta}$, we conclude that $\left|s_A\right|_C \leq 1$.
		The error $x$ results from $t$ circuit faults, so its support must be contained in the support of at most $t$ rows $\left\{s_i\right\}_{i \in [t]}$ of $H_X'$ i.e. $\supp x \subseteq \cup_{i \in [t]}\supp s_i$.
		Then, 
		\begin{align*}
			\left|x_A\right|_C &\leq \left|\cup_{i \in [t]}\supp \left(s_i\right)_A\right|_C\\
						&\leq \sum_{i \in [t]}\left|\left(s_i\right)_A\right|_C \\
						&\leq t.
		\end{align*}

		We now present the proof for $\left|z_A\right|_R \leq t \cdot \left\lfloor \frac{w_Z}{2} \right\rfloor$.
		Consider some $\vec{\gamma} \in \F_2^{n_Z}, \vec{\gamma}' \in \F_2^{n}, \vec{\zeta}, \in \F_2^{n_c},  \vec{\zeta}' \in \F_2^{n_c - k_c}$ with $|\vec{\gamma}| = |\vec{\zeta}|  = |\vec{\gamma}'| = |\vec{\zeta}'| = 1$.
		Let $\vec{v}$ be a single row of $H_Z'$ from the top partition stated in \eqref{eq:parity_check_matrix_thickened_Z} such that
		\begin{align}
			\vec{v} = H_Z^{\prime\top}\left(\begin{array}{c}\vec{\gamma} \otimes \vec{\zeta} \\ \textbf{0} \otimes \textbf{0}\end{array}\right) = H_Z^{\top}\vec{\gamma} \otimes \vec{\zeta} = \vec{v}_A.
		\end{align}
		Using the same analysis and matrix rearrangement as above as well as Lemma~\ref{lem:elementary_fault_propagation}, we conclude that $\left|v_A\right|_R \leq \left\lfloor \frac{w_Z}{2}\right\rfloor$.
		Let $\vec{v}'$ be a single row of $H_Z'$ from the bottom partition stated in \eqref{eq:parity_check_matrix_thickened_Z} such that
		\begin{align}
			\vec{v}' = H_Z^{\prime\top}\left(\begin{array}{c}\textbf{0} \otimes \textbf{0}\\ \vec{\gamma}' \otimes \vec{\zeta}' \end{array}\right) = \left(\vec{\gamma}' \otimes H_C^{\top} \vec{\zeta}'\right) \oplus \left(H_X\vec{\gamma}' \otimes \vec{\zeta}'\right) = \vec{v}'_A + \vec{v}'_B.
		\end{align}
		Using the same analysis and the fact that $\left|\vec{\gamma}'\right| = 1$, we can conclude that $\left|v'_A\right|_R  \leq 1$.
		Then, the error $z$ that results from $t$ circuit faults has the following:
		\begin{align}
			\left|z_A\right|_R \leq t\cdot \left\lfloor \frac{w_Z}{2}\right\rfloor.
		\end{align}		
		
		The proof for $\left|z_A\right|_C \leq t\cdot \left\lfloor\frac{w_C}{2}\right\rfloor$ is similar to the proof above.
\end{proof}

We now state an important lemma regarding the algebraic description for the logical operators of the generalized thickened quantum code.
Our proofs for the logical operators of $\calQ'$ are largely similar to the proof of Lemma~1 in \cite{krishna2021fault}.

\begin{lemma}[Logical Operators of Generalized Thickened Quantum Code]\label{lem:logical_operators_thickened}
	The $X$ logical operators of $\calQ'$ are spanned by
		\begin{align}&\left(\frac{\ker H_Z}{\rs H_X} \otimes \ker H_C \;\middle|\; \textbf{0}\right)
		\end{align}
	where the left partition corresponds to the qubits in region $A$ and the right partition corresponds to the qubits in region $B$.
	By the same notation, the $Z$ logical operators of $\calQ'$ are spanned by
	\begin{align}
		&\left(\frac{\ker H_X}{\rs H_Z} \otimes \frac{\F_2^{n_c}}{\rs H_C} \;\middle|\; \textbf{0}\right).
		\end{align}
\end{lemma}
\begin{proof}
	Let $\vec{x}$ be an $X$ logical operator, i.e., 
	\[\vec{x} \in \left(\frac{\ker H_Z}{\rs H_X} \otimes \ker H_C \;\middle|\; \textbf{0}\right).\]
	It is apparent that $\vec{x}$ commutes with $\calQ'$'s $Z$ stabilizers.
	Suppose that $\vec{x}$ is in fact in the span of the $X$ stabilizers, i.e., there exists a non-trivial vector $\vec{a} \in \F_2^{n_X n_c}$ such that $H_X^{\prime \top} \vec{a} = \vec{x}$ for the sake of contradiction.
	It is possible to assume that $\supp \vec{x}_A \neq \emptyset$ without loss of generality which gives us $\left(H_X^\top \otimes \mathbbm{1}_{n_c}\right)\vec{a} = \vec{x}_A.$
	
	Note that there exists some column $c$ such that $x_A[\cdot, \{c\}] \neq \mathbf{0}$ since $x_A$ is non-trivial by assumption.
	We can interpret $x_A[\cdot, \{c\}]$ as the support of the $X$ logical operator on all qubits in region $A$ that correspond to some column $c$ in the schematic of the generalized thickened quantum code $\calQ'$ shown in Figure~\ref{fig:schematic}.
	Similarly, we can let $a\left[\cdot, \{c\}\right]$ denote the vector over $\F_2^{n_X}$ as the vector that selects the $X$ stabilizer generators that has support on the qubits in region $A$ and column $c$.
	Then, we have	 $H_X^\top\vec{a}\left[\cdot, \{c\}\right] = \vec{x}_A[\cdot, \{c\}].$
	However, $\vec{x}_A[\cdot,\{c\}] \in \ker H_Z / \rs H_X$ and thus does not lie in the row space of $H_X$.
	We have arrived at a contradiction and $\vec{x}$ is not in the span of the $X$ stabilizers of $\calQ'$.

	All that remains is to show that the $X$ logical operators of $\calQ'$ are spanned by the set of logical operators given in the lemma statement.
	The number of elements in $\left(\ker H_Z / \rs H_X\right) \otimes \ker H_C$ is $k \cdot k_c$ which matches the number of logical qubits for $\calQ'$ stated in Theorem~\ref{thm:generalized_thickening}
	Thus, we have shown that space spanned by the $X$ logical operators of $\calQ'$ stated in the lemma is indeed the space of all $X$ logical operators of $\calQ'$.
	The proof for the $Z$ logical operators of $\calQ'$ is similar and is omitted for brevity.
 \end{proof}

In the following lemma, we take inspiration from the proofs in Refs.~\cite{tillich2013quantum,zeng2019higher,manes2023distance} to characterize the logical operators of the generalized thickened quantum code.
In particular, we show that any non-trivial $X/Z$ logical operator of the generalized thickened quantum code has to be supported on a certain number of rows and columns of qubits in the schematic of the generalized thickened quantum code.
This provides a lower bound on the number of rows and columns of qubits that a logical operator of the generalized thickened quantum code has to be supported on.

\begin{lemma}[Generalized Thickened Quantum Code Component Weight Distance]\label{lem:thickened_code_component_weight}
	For any $x,z \in \F_2^{n \times n_c} \oplus \F_2^{n_X \times \left(n_c - k_c\right)}$:
	\begin{enumerate}
		\item if $\vec{x} \in \ker H_Z' \setminus \rs H_X'$, then $\left|\supp x_A\right| \geq d_X \cdot d_c$ with $\left|x_A\right|_C \geq d_c$ and $\left|x_A\right|_R \geq d_X$.
		\item if $\vec{z} \in \ker H_X'\setminus \rs H_Z'$, then $\left|\supp z_A\right| \geq d_Z$ with $\left|z_A\right|_R \geq d_Z$.
	\end{enumerate}
\end{lemma}
\begin{proof}
	We begin by proving the first statement.
	From Lemma~\ref{lem:logical_operators_thickened}, the $X$ logical operators of the generalized thickened quantum code $\calQ'$ are spanned by $\left(\frac{\ker H_Z}{\rs H_X} \otimes \ker H_C \;\middle|\; \textbf{0}\right)$.
	It is clear from the above expression that any combination of $X$ logical operators (without any $X$ stabilizers) will have support on at least $d_X \cdot d_c$ qubits in region $A$ because $d_X' = d_X \cdot d_c$.
	Now, consider $x \in \F_2^{n \times n_c} \oplus \F_2^{n_X \times \left(n_c - k_c\right)}$ such that $\vec{x} = \sum_i \vec{x}_i \otimes \vec{v}_i$ where each $\vec{x}_i$ is some non-trivial $X$ logical operator representative of $\calQ$ and $\vec{v}_i$ is some codeword of the classical code $\calC$.
	Then, $x_A$ has weight at least $d_X \cdot d_c$.
	Let $\left(\vec{s}_x\right)_A = \sum_{j} \vec{s}_{x_{j}} \otimes \vec{c}_j$ be some arbitrary combination of $X$ stabilizer generators of $\calQ'$ with support restricted to the region $A$ where $\vec{s}_{x_j} \in \row\left(H_X\right)$ and $\vec{c}_j$ is some unit vector of dimension $n_c$ with the sum modulo 2 running over the arbitrary subset of the $X$ stabilizer generators of $\calQ'$.
	
	Because $\vec{x} \in \ker H_Z'$, we obtain $\left|x_A\right|_C \geq d_c$ since the non-trivial codewords of $\calC$ lies in the kernel of $H_C$ and have weight at least $d_c$.
	Thus, there exists at least $d_c$ different unit vectors of dimension $n_c$, i.e., $\vec{c}_k \in \F_2^{n_c}$, such that $\left\langle\vec{v}_i , \vec{c}_k\right\rangle = 1$ for some values of $i$.
	Then, for any one of these $\vec{c}_k$, we have that
	\begin{align}
		\left|\left(\mathbbm{1}_{n} \otimes \vec{c}_k^\top\right)\left(\left(\vec{s}_x\right)_A + \vec{x}_A\right)\right| &= \left|\left(\mathbbm{1}_{n} \otimes \vec{c}_k^\top\right)\left(\sum_{j} \vec{s}_{x_{j}} \otimes \vec{c}_j + \sum_i \vec{x}_i \otimes \vec{v}_i\right)\right| \\
		&= \left|\sum_{j:\;\left\langle \vec{c}_j, \vec{c}_k\right\rangle = 1} \vec{s}_{x_j} + \sum_{i:\;\left\langle \vec{v}_i,\vec{c}_k\right\rangle = 1} \vec{x}_i\right| \\
		&\geq d_X.
	\end{align}
	The application of $\left(\mathbbm{1}_{n} \otimes \vec{c}_k^\top\right)$ can be interpreted as us selecting the qubits in the column corresponding to $\vec{c}_k$ in the region $A$ that lie in the support of our $X$ logical operator.
	The last inequality comes from the fact that any combination of $X$ logical operators and $X$ stabilizers of $\calQ$ has weight at least $d_X$.
	In other words, for at least $d_c$ different columns $k \in \left[n_c\right]$ in the region $A$, $x_A + \left(s_X\right)_A$ is supported on at least $d_X$ qubits.
	This implies that $\left|x_A + \left(s_X\right)_A\right|_R \geq d_X$ and $\left|x_A + \left(s_X\right)_A\right|_C \geq d_c$ as well as $\left|\supp\left(x_A + \left(s_X\right)_A\right)\right| \geq d_X \cdot d_c$.
	Thus, we have shown the first statement of the lemma.

	We now prove the second statement of the lemma.
	Our goal is to show that if $\vec{z} \in \ker H_X'$ and $\left|z_A\right|_R < d_Z$, then $\vec{z}$ lies in $\rs H'_Z$.
	Let $R_A \subseteq [n]$ be the set of rows in region $A$, i.e., $\F_2^{n \times n_c}$, such that for any $r \in R_A$, there exists some $c \in [n_c]$ such that $z_A[\{r\},\{c\}] = 1$.
	In other words, $R_A$ is the set of rows in region $A$ where $z_A$ has non-trivial support.
	Then, we have that $\left|z_A\right|_R = |R_A|$.
	Likewise, let $C_B \subseteq \left[n_c - k_c\right]$ be the set of columns in region $B$, i.e., $\F_2^{n_X \times \left(n_c - k_c\right)}$, such that for any $c' \in C_B$, there exists some $r' \in [n_X]$ such that $z_B[\{r'\},\{c'\}] = 1$.

	Let $H_X|_{R_A}$ denote the submatrix of $H_X$ with columns restricted to the qubits indexed by $R_A$.
	In other words, $H_X|_{R_A}$ is $H_X$ restricted to the qubits in the original quantum code $\calQ$ that correspond to the rows in region $A$ where $z_A$ has non-trivial support, i.e.,
	\[\row\left(H_X|_{R_A}\right) = \left\{\vec{r}\left[R_A\right] \in \F_2^{\left|R_A\right|}\;\middle|\; \vec{r} \in \row\left(H_X\right)\right\}\]
	where $\vec{r}\left[R_A\right]$ denotes the restriction of $\vec{r} \in \F_2^{n}$ to the indices in $R_A$.
	On the other hand, let $H_Z|^{R_A}$ denote the submatrix of $H_Z$ such that
	\[\row\left(H_Z|^{R_A}\right) = \left\{\vec{r}\left[R_A\right] \in \F_2^{\left|R_A\right|}\;\middle|\; \vec{r} \in \row\left(H_Z\right), \vec{r}\left[[n] \setminus R_A\right] = \textbf{0}\right\}.\]
	In other words, $H_Z|^{R_A}$ is $H_Z$ without rows that have non-trivial support outside of $R_A$ and restricted to the columns corresponding to qubits indexed by $R_A$.
	Now, we denote $\calQ\left(R_A\right)$ to be the quantum code with stabilizer matrices $H_X|_{R_A}$ and $H_Z|^{R_A}$.
	At the same time, let $H_C^\top|_{C_B}$ denote the submatrix of $H_C^\top$ with columns restricted to the classical checks indexed by $C_B$.
	From our assumption that $\left|z_A\right|_R < d_Z$, we have that $\frac{\ker H_X|_{R_A}}{\rs H_Z|^{R_A}}$ is trivial which implies that $\calQ\left(R_A\right)$ encodes no logical qubit.
	By Theorem~\ref{thm:generalized_thickening}, the generalized thickened quantum code $\calQ\left(H_X|_{R_A}, H_Z|^{R_A}, H_C|_{C_B}\right)$, which we denote as $\calQ^{\ast}$ for brevity, has no logical qubit.
	Then, denoting $H_X^\ast$ and $H_Z^\ast$ as the stabilizer matrices of $\calQ^{\ast}$ that can be obtained by adapting \eqref{eq:parity_check_matrix_thickened_X} and \eqref{eq:parity_check_matrix_thickened_Z}, we conclude that $\ker H_X^\ast = \rs H_Z^\ast$.
	
	Let $z|_{R_A \cup C_B}$ be the following:
	\begin{align}
		z|_{R_A \cup C_B} = z_A|_{R_A} \oplus z_B|_{C_B}
	\end{align}
	where $z|_{R_A \cup C_B}$ is the restriction of $z$ to the qubits in the rows $R_A$ of the region $A$ and qubits in the columns $C_B$ of the region $B$.
	Since the non-trivial components of $\vec{z}|_{R_A \cup C_B}$ lie completely in the domain of $H_X^\ast$ and $H_Z^\ast$, $\vec{z}|_{R_A \cup C_B} \in \ker H_X^\ast$ and $\vec{z}|_{R_A \cup C_B} \in \rs H_Z^\ast$.
	In other words, there exists some $\vec{a} \in \F_2^{n_c\left|\row\left(H_Z|^{R_A}\right)\right| + \left|R_A\right|\left|C_B\right|}$ such that $\vec{z}|_{R_A \cup C_B} = H_Z^{\ast \top} \vec{a}$.
	Since $H_Z^\ast$ is formed by a restriction of the domain of $H_Z'$, i.e., the stabilizer matrix of the thickened quantum code $\calQ'$, there exists some $a' \in \F_2^{n_z n_c + n\left(n_c - k_c\right)}$ such that $\vec{z} = H_Z'^{\top} \vec{a}'$.
	In fact, $\vec{a}'$ can be obtained from $\vec{a}$ by appending zeros to the indices that correspond to stabilizer generators that were not present in $H_Z^\ast$.
	Thus, $\vec{z} \in \rs H_Z'$.

	\end{proof}

	Now, we are ready to prove the effective distance preservation of the generalized thickened quantum code.
	We show that the effective distance of the generalized thickened quantum code is preserved under the thickening operation by combining the upper bound on the data qubit errors that emerge as a result of hook errors and the lower bound on the component weight of the logical operators of the generalized thickened quantum code.

	\begin{theorem}[Effective Distance of Generalized Thickened Quantum Code]\label{thm:effective_distance_preservation_thickened_code}
		Let $Q\left(H_X, H_Z\right)$ be a quantum code with distances $d_X$ and $d_Z$ that has effective distances $\overline{d}^{M}_X$ and $\overline{d}^{M}_Z$ for some single ancilla stabilizer measurement schedule $M$.
		Let $Q'\left(H_X, H_Z, H_C\right)$ be the generalized thickened quantum code that is the result of a product between $Q\left(H_X, H_Z\right)$ and the classical code $\calC = \left[n_c, k_c, d_c\right]$ with parity check matrix $H_C$.
		Then, there exists some single ancilla syndrome measurement schedule $M'$ such that the effective distances of the generalized thickened quantum code $Q'\left(H_X, H_Z, H_C\right)$ is given by $\overline{d'}^{M'}_X = \overline{d}^{M}_X \cdot d_c$ and $\overline{d'}^{M'}_Z = \overline{d}^{M}_Z$. 
		\end{theorem}
	\begin{proof}
		Let $\hat{i}$ denotes the unit vector of dimension $n_c$ with 1 at its $i\textsuperscript{th}$ index.
		The single ancilla syndrome measurement schedule $M'$ can be constructed as such:
		\begin{enumerate}
			\item Iterate through the given schedule $M$ for $Q$.
			\begin{enumerate}
				\item If the current stabilizer generator being iterated through is some $s_X$ that corresponds to some row of $H_X$, then for all $i \in [n_c]$, measure the qubits $s_X \otimes \hat{i}$ in the region $A$ in the same order as how it would be measured in $M$.
				Then, measure the corresponding new qubits in the region $B$ in any order.
				\item If the current stabilizer generator being iterated through is some $s_Z$ that corresponds to some row of $H_Z$, then for all $i \in [n_c]$, measure the qubits $s_Z \otimes \hat{i}$ in the region $A$ in the same order as how it would be measured in $M$.
			\end{enumerate}
			\item Measure the $Z$ stabilizer generators that correspond to the rows of $\left(\mathbbm{1}_n \otimes H_C\;\middle|\; H_X^{\top} \otimes \mathbbm{1}_{n_c - k_c}\right)$ in any order.
		\end{enumerate}

		From Lemma~\ref{lem:thickened_code_component_weight}, we know that any $X$ logical operator $\vec{x}$ of $Q'$ has $\left|\supp x_A\right| \geq  d_c \cdot d_X$ with $\left|x_A\right|_C \geq  d_c$ and $\left|x_A\right|_R \geq d_X$.
		Define $x_{A[\cdot, i]}$ to be the restriction of $x_A$ to the $i$th column of region $A$.
		Since $\left|\supp x_A\right| \geq d_c\cdot d_X$ and $\left|x_A\right|_C \geq d_c$, $\left|x_{A[\cdot, i]}\right| \geq d_X$ for at least $d_c$ different $i \in [n_c]$.
		By Proposition~\ref{prop:elementary_faults_bounded_component_weight}, we know that $\left|\left(s_X\right)_A\right|_C = 1$ for each $X$ stabilizer generator $s_X$ that lies in the row of $H_X \otimes \mathbbm{1}_{n_c}$.
		Thus, we can analyze the elementary faults required to generate $x_{A[\cdot, i]}$ separately for each $i \in [n_c]$.
		Consider the subset of $X$ stabilizer generators of $Q'$ that are described by the rows that correspond to $H_X \otimes \hat{i}$.
		From the fact that the effective distance of $Q$ is $\overline{d}^M_X$ given $M$, we know that the number of elementary faults on  qubits in $A\left[\cdot, i\right]$ and the subset of $X$ stabilizer generators required to generate $x_{A[\cdot, i]}$ is at least $\overline{d}^M_X$. 
		Since we have at least $d_c$ such values of $i$ (columns), the total number of elementary faults required to generate $x_A$ is at least $d_c \cdot \overline{d}^M_X$ based on the schedule $M'$.
		Thus, we have shown that the effective distance of the thickened quantum code $Q'$ is at least $d_c \cdot \overline{d}^M_X$.
		Since the effective distance is upper bounded by the actual distance of the code, we have that $\overline{d'}^{M'}_X = d_c \cdot \overline{d}^M_X$.

		We can show that $\overline{d'}^{M'}_Z = \overline{d}^M_Z$ in a similar way.
		From Lemma~\ref{lem:thickened_code_component_weight}, we know that any $Z$ logical operator $\vec{z}$ of $Q'$ has $\left|\supp z_A\right| \geq d_Z$ with $\left|z_A\right|_R \geq d_Z$.
		By Proposition~\ref{prop:elementary_faults_bounded_component_weight}, we know that $\left|\left(s_{Z[B]}\right)_A\right|_R = 1$ for each $Z$ stabilizer generator of $s_{Z[B]}$, i.e., corresponding to the rows of $\left(\mathbbm{1}_n \otimes H_C \;\middle|\; H_X^{\top} \otimes \mathbbm{1}_{n_c - k_c}\right)$.
		Thus, single qubit errors and hook errors from the $Z$ stabilizer generators of $s_{Z[B]}$ are functionally the same in the sense that each fault impacts at most a single row and we need at least $d_Z$ of such elementary faults to generate $z_A$.
		However, recall that the effective distance of $Q$ is $\overline{d}^M_Z$ given $M$ and $\left|\left(s_{Z[T]}\right)_A\right|_C = 1$ for each $Z$ stabilizer generator of $s_{Z[T]}$, i.e., corresponding to the rows of $\left(H_Z \otimes \mathbbm{1}_{n_c} \;\middle|\; \textbf{0}\right)$.
		This implies that at least $\overline{d}^M_Z$ elementary faults are required on qubits in $A$ and $Z$ stabilizer generators in $s_{Z[T]}$ using schedule $M'$ to form a $Z$ operator that has support on at least $d_Z$ rows.
		Since $\overline{d}^M_Z$ is upper bounded by $d_Z$, we have shown that $\overline{d'}_Z^{M'} = \overline{d}^M_Z$ for the thickened quantum code $Q'$ when considering all $Z$ stabilizer generators of $Q'$ and single qubit errors.
	\end{proof}

Because Theorem~\ref{thm:effective_distance_preservation_thickened_code} holds for generalized thickening, it also holds for the specific thickening technique constructed by Hastings using a classical $[\ell, 1, \ell]$ repetition code.
While the theorem does not explicitly mention how choosing heights might affect the effective distance of the thickened code $\calQ'$, it is relatively straightforward to see how the distance gained from thickening is preserved by the tensor product structure of the thickened code even after choosing heights.
Because choosing heights removes redundant $Z$ stabilizer generators in $\left(H_Z \otimes \mathbbm{1}_\ell\;\middle|\;\textbf{0}\right)$, the new set of possible hook errors from the remaining stabilizer generators has to be contained within the original set of possible hook errors.
Thus, the effective distance of the thickened code $\calQ'$ cannot decrease after choosing heights. 

\begin{figure}[h]
	\centering
	\begin{tikzpicture}
		\draw (0,0) rectangle (3,3); 
		\node at (1.5,1.5) {Redundant};
		\draw (3.4,3.4) rectangle (6.4,6.4); 
		\node at (4.9,4.9) {$A$ (qubits)};
		\draw (3.4,0) rectangle (6.4,3); 
		\node at (4.9,1.5) {$X$ stabilizers};
		\draw (3.4,6.8) rectangle (6.4,9.8); 
		\node at (4.9,8.3) {$Z$ stabilizers};
		\draw (0,3.4) rectangle (3,6.4); 
		\node at (1.5,4.9) {$X$ stabilizers};
		\draw (0,6.8) rectangle (3,9.8); 
		\node at (1.5,8.3) {$B$ (qubits)};
		\draw (-0.5,0) edge[-] (-0.5,3); 
		\node at (-0.8,1.5) {$n_X$};
		\draw (-0.5,3.4) edge[-] (-0.5,6.4); 
		\node at (-0.8,4.9) {$n$};
		\draw (-0.5,6.8) edge[-] (-0.5,9.8); 
		\node at (-0.8,8.3) {$n_Z$};
		\draw (0,-0.5) edge[-] (3,-0.5); 
		\node at (1.5,-0.75) {$n_c$};
		\draw (3.4,-0.5) edge[-] (6.4,-0.5); 
		\node at (4.9,-0.75) {$n_c - k_c$};
	\end{tikzpicture}
	\caption{Schematic for the generalized thickened quantum code when $d_Z$ increases instead of $d_X$. 
	There are $n_c - k_c$ copies of the quantum code $\calQ$ arranged in $n_c - k_c$ columns in the three blocks on the right column.
	For each of the $n_c - k_c$ columns, we have the $Z$ stabilizer generators and $X$ stabilizer generators belonging to that column acting on the qubits in the same column.
	Similarly, there are $n_Z$ copies of the classical code $\calC$ arranged in $n_Z$ rows in the two blocks in the top row.
	For each of the $n_Z$ rows, we have the $Z$ stabilizer generators belonging to that row acting on the qubits in the same row which is reminiscent to how the checks of the classical code act on the classical bits.
	For the qubits in the region $B$, we have the $X$ stabilizer generators from the middle-left block and $Z$ stabilizer generators acting on the qubits in the same column and row respectively. 
	}
 \label{fig:alt_schematic}
\end{figure}

In addition, we note that Theorem~\ref{thm:effective_distance_preservation_thickened_code} can be easily adapted to show that the effective $Z$ distance is ``thickened'' and preserved when we perform the following dual mapping instead:
	\begin{center}
		\begin{tabular}{|c|c|c|}
			\hline
			Vector Spaces & Original Theorem & Adapted Theorem \\
			\hline
			$\F_2^{n_X} \otimes \F_2^{n_c}$ & $X$ stabilizers & undefined \\
			\hline
			$\left(\F_2^{n} \otimes \F_2^{n_c}\right) \oplus \left(\F_2^{n_X}\otimes \F_2^{n_c - k_c}\right)$ & qubits  & X stabilizers \\
			\hline
			$\left(\F_2^{n_Z} \otimes \F_2^{n_c}\right) \oplus \left(\F_2^{n}\otimes \F_2^{n_c - k_c}\right)$ & Z stabilizers  & qubits \\
			\hline
			$\F_2^{n_Z} \otimes \F_2^{n_c - k_c}$ & undefined & Z stabilizers \\
			\hline
		\end{tabular}
	\end{center}
	We provide a geometric picture of the generalized thickening that is done in a different ``basis'' in Fig.~\ref{fig:alt_schematic} to help with the intuition.
	By pegging the stabilizer generators and qubits to the cells that are one-dimensional higher, we can observe that the $Z$ distance is now thickened instead due to the additional $X$ stabilizer generators and a similar argument to Theorem~\ref{thm:effective_distance_preservation_thickened_code} can be made to show that the gain in effective $Z$ distance is preserved.

\subsection{Distance Preservation of Higher-dimensional Hypergraph Product Codes}
Theorem~\ref{thm:effective_distance_preservation_thickened_code} allows us to make another observation regarding higher-dimensional HGP codes.
Before we get to the observation, we first restate a theorem that detail the distances of HGP codes as well as a recent result that shows the robustness of the HGP code's effective distance with respect to single-ancilla syndrome extraction circuits.

\begin{theorem}
	[{\cite[Distances of HGP Codes]{tillich2013quantum}}]\label{thm:distance_HGP}
	For $i \in \{1, 2\}$, let $d_i$ and $d_i^\top$ be the minimum distances of classical codes $\calC_i$ and $\calC_i^\top$ with check matrices $H_{C_i}$ and $H_{C_i}^\top$ respectively. 
	Denote $d_X$ and $d_Z$ to be the the $X$ and $Z$ distance of the HGP code $\calQ$ that is a result of the tensor product of the complexes corresponding to $\calC_1$ and $\calC_2^\top$.
	Using the convention that the distance of a code that only has the all-zero codeword is $\infty$, we have 
	\begin{align}
		&d_X \geq \min\left(d_1, d_2\right)\quad\text{with}\quad d_X=\begin{cases}
		d_1 & \text{if}\quad d_1 \leq d_2\;\text{and}\;d_2^\top \neq \infty, \\
		d_2 & \text{if}\quad d_1 \geq d_2\;\text{and}\;d_1^\top \neq \infty, \\
		\end{cases}  \\
		&d_Z \geq \min\left(d_1^\top, d_2^\top\right) \quad\text{with}\quad d_Z=\begin{cases}
		d_1^\top & \text{if}\quad d_1^\top \leq d_2^\top\;\text{and}\;d_2 \neq \infty, \\
		d_2^\top & \text{if}\quad d_1^\top \geq d_2^\top\;\text{and}\;d_1 \neq \infty. \\
		\end{cases}
	\end{align}
\end{theorem}

\begin{theorem}[{\cite[Distance-Preserving HGP Codes]{manes2023distance}}]\label{thm:distance_preserving_HGP}
	Let $d$ be the distance of an HGP code.
	The effective distance of the code using any stabilizer measurement circuit is also $d$.
\end{theorem}

The main idea behind the proof of the above Theorem~\ref{thm:distance_preserving_HGP} is that an $X$ stabilizer of the HGP code acts only on the bit-type qubits in the same column as the stabilizer as well as the check-type qubits in the same row as the stabilizer.
Therefore, a single error in the stabilizer measurement circuit can only affect those qubits in the same row and column as the stabilizer.
However, every non-trivial logical $X$ operator has support in at least $d$ distinct rows or $d$ distinct columns.
Thus, $d$ elementary faults are needed to cause an $X$ logical error in the HGP code.
The same argument can be made for the non-trivial logical $Z$ operators of the HGP code.

In 2019, Zeng and Pryadko provided a generalization of HGP codes that they term higher-dimensional HGP codes~\cite{zeng2019higher}.
These codes belong to the homological product code family~\cite{bravyi2014homological} that can be decomposed into 1-complexes.
They form $m$-complexes for $m \geq 2$ and also encapsulate all families of toric codes on hypercubic lattices that are $m$-dimensional.
These $m$-complexes can be recursively constructed by repeatedly tensor producting (or thickening) an $(m-1)$-complex with a 1-complex that corresponds to a classical code.
Using Theorem~\ref{thm:effective_distance_preservation_thickened_code} and Theorem~\ref{thm:distance_preserving_HGP}, we can make the following observation regarding higher-dimensional HGP codes using an inductive argument:

\begin{theorem}[Effective Distance of Higher-Dimensional HGP Codes]
	Let $\calQ$ be a $\left[\left[n, k, d_X, d_Z\right]\right]$ higher-dimensional HGP code that has a correspondence to a $D$-complex $\calA$ for some $D \in \N$ such that $D \geq 2$.
	In particular, let the $D$-complex $\calA$ be
	\begin{align}
		\bigotimes_{i \in \set{1, \ldots, D}} \calB_i
	\end{align} 
	where each $\calB_i$ are $1$-complexes with boundary and co-boundary operators $\partial^{\calB_i}_0, \partial^{\calB_i}_1, \delta^{\calB_i}_0, \delta^{\calB_i}_1$.
	Without loss of generality, we assume that $k = \dim \calH_a(\calA)$ for some $a \in \N$ and $1 \leq a \leq D - 1$.
	In other words, the $X$ stabilizer generators, qubits, and $Z$ stabilizer generators of $\calQ$ correspond to the basis elements of the vector spaces $A_{a-1}, A_a, A_{a+1}$ of $\calA$ respectively.
	The quantum code $\calQ$ has the parity check matrices:
	\[H_X = \partial_a^{\calA}\quad\text{and}\quad H_Z = \delta_a^{\calA}.\]
	These $\calB_i$s correspond to the classical $\left[n_c^{(i)}, k_c^{(i)}, d_c^{(i)}\right]$ codes $\calC_i$ with parity check matrices $H_C^{(i)}$ if $i \in [a]$ and $H_C^{(i)\top}$ otherwise such that each $H_C^{(i)}$ is some full-row rank $r_i \times c_i$ binary matrix with $r_i < c_i$.
	Then, the effective distances of $\calQ$ using any single-ancilla stabilizer measurement circuit are also $d_X$ and $d_Z$.
\end{theorem}

\begin{proof}
	Let us decompose the $D$-complex $\calA$ into the following two complexes:
	\begin{align}
		\calA &= \underbrace{\left(\bigotimes_{i = 1}^{a - 1} \calB_i\right)}_{\calB_L} \otimes \underbrace{\left(\bigotimes_{j = a}^{D} \calB_j\right)}_{\calB_R}.
	\end{align}
	Since $\calB_R$ is a complex with at least $3$ terms given $a \leq D-2$, we can consider it as the following quantum code:
	\[\calQ_R\left(H_X = \partial_1^{\calB_R}, H_Z = \delta_1^{\calB_R}\right).\]
	By the associativity of the tensor product of chain complexes, we can express $\calB_R$ as 
	\begin{align}
		\left(\ldots\left(\left(\calB_a \otimes \calB_{a+1}\right) \otimes \calB_{a+2}\right) \ldots \otimes \calB_D\right). \label{eq:calB_R}
	\end{align}

	Let us denote the $X$ distance and $Z$ distance of $\calQ_R$ as $d_X^R$ and $d_Z^R$ respectively.
	By Theorem~\ref{thm:thickening_distance} and \eqref{eq:calB_R},  we have
	\begin{align}
		d_X^R \geq \min\left(d_a, d_{a+1}\right)\prod_{j = a+2}^{D}d_j \quad\text{and}\quad d_Z^R \geq \min\left(d_{a}^\top, d_{a+1}^\top\right) \label{eq:distance_calB_R}
	\end{align}
	where $d_a^\top$ and $d_{a+1}^\top$ are the distances of the classical codes that correspond to the co-complexes of $\calB_a$ and $\calB_{a+1}$ respectively.
	\eqref{eq:distance_calB_R} can be shown easily via induction.
	For the base case where $D = a + 1$, the statement is true from Theorem~\ref{thm:distance_HGP}.
	For the inductive step, suppose the statement is true for $D = k$ for some $k \in \N$ and $k \geq 2$.
	Then, by Theorem~\ref{thm:thickening_distance}, the induction step can be shown to be true for $D = k + 1$.
	The same induction argument can be applied for the effective distances of $\calQ_R$.
	The base case and inductive step can be shown to be true by Theorems~\ref{thm:distance_preserving_HGP}~and~\ref{thm:effective_distance_preservation_thickened_code} respectively.

	By using the associativity of the tensor product of chain complexes again, we can express $\calA$ as
	\begin{align}
		\calA = \left(\calB_1 \otimes \ldots \left(\calB_{a-2} \otimes \left(\calB_{a-1} \otimes \calB_R\right)\right)\ldots\right). \label{eq:calA}
	\end{align}
	Using the same inductive argument as before as well as Theorem~\ref{thm:thickening_distance} and \eqref{eq:calA}, we have the following for the distances of $\calQ$:
	\begin{align}
		d_X = d_X^R \quad\text{and}\quad d_Z = \left(\prod_{i = 1}^{a-1}d_i\right) \cdot d_Z^R. \label{eq:distance_calA}
	\end{align}
	To complete our proof for the effective distances of $\calQ$, we can use the same inductive argument as before including Theorem~\ref{thm:effective_distance_preservation_thickened_code} and \eqref{eq:calA}.
	By using the adapted version of Theorem~\ref{thm:effective_distance_preservation_thickened_code} for the generalized thickening of the $Z$ distance, we can inductively show that the effective distances of $\calQ$ are preserved under any single-ancilla stabilizer measurement circuit.
	\end{proof}

	Another way to interpret the above theorem is that the complex $\calB_L$ serves to ``lift'' the base HGP code formed by the tensor product of $\calB_{a}$ and $\calB_{a+1}$ to the $a$\textsuperscript{th} dimension.
	This process helps to thicken the $Z$ distance and also preserves the increase in $Z$ distance.
	Then, the subsequent tensor products with the remaining $\calB_i$s help to thicken the $X$ distance and preserve the increase in $X$ distance.

\subsection{Distance Preservation of Reduced Cone Codes}
\label{sec:distance_preservation_reduced_cone}
In this section, we will show that the effective distance of reduced cone codes is approximately preserved under single-ancilla syndrome extraction circuits.
In particular, we construct a single-ancilla stabilizer measurement schedule for reduced cone codes that perfectly preserves the effective $X$ distance and preserves the $Z$ distance up to a factor of $\Omega\left(1/q_X\right)$.

\begin{theorem}[Partial Effective Distance Preservation for Reduced Cone Codes]
	Let $\calQ$ be a quantum code with distances $d_X$ and $d_Z$ that has effective distances $\overline{d}^M_X$ and $\overline{d}^M_Z$ for some single-ancilla stabilizer measurement schedule $M$.
	Suppose we let $\calQ'$ be the reduced cone code that we obtain after applying coning to all $Z$ stabilizer generators with weight larger than 5.
	The new distances of $\calQ'$ are $d_X' = d_X$ and $d_Z' \geq d_Z \lambda \ell$ where $\lambda$ is the soundness factor defined in Ref.~\cite{hastings2021quantum} and $\ell$ is the thickness factor introduced by thickening in the cone reduction step.
	Then, there exists a single-ancilla stabilizer emasurement schedule $M'$ for $\calQ'$ such that the effective distances of $\calQ'$ are $\overline{d'}^{M'}_X = \overline{d}^M_X$ and $\overline{d'}^{M'}_Z \geq \Omega\left(1/q_X\right) \overline{d}^M_Z \lambda \ell$.
\end{theorem}
\begin{proof}
	To prove the theorem statement, we break the proof down into three parts.
	The three different parts will correspond to the codes resulting from the three different steps in Hastings's construction of coning: cone code, thickened and height-chosen cone code, and reduced cone code.
	In the first part, we adapt the measurement schedule $M$ for $\calQ$ to a measurement schedule $M_C$ for the cone code $\Cone\left(\calQ, f\right)$.
	In the second part, we adapt the measurement schedule $M_C$ for the cone code to a measurement schedule $M_{\text{thick}}$ for the thickened and height-chosen cone code $\Cone\left(\calQ, f, \ell\right)$.
	In the final part, we adapt the measurement schedule $M_{\text{thick}}$ for the thickened and height-chosen cone code to a measurement schedule $M'$ for the reduced cone code post-cellulation $\calQ' = \Cone_R\left(\calQ, f, \ell\right)$.

	\textbf{Part 1: Adaptation of $M$ to $M_C$ for the Cone Code.}
	Consider the following stabilizer measurement schedule $M_C$ for the cone code $\Cone\left(\calQ, f\right)$:
	\begin{enumerate}
		\item Iterate through the given schedule $M$ for $\calQ$.
		\begin{enumerate}
			\item If the current stabilizer generator being iterated through is some $s_X$ that corresponds to some row of $H_X$, then entangle the qubits in the support of $s_X$ in the same order as how it would be measured in $M$.
				If the $s_X$ corresponds to a $X$ stabilizer generator in the cone code that now has support on a qubit in the cone, i.e., a 0-cell in some $\overline{\calB}_i$, then entangle that additional qubit in the cone code at the end after entangling all the original qubits.
			\item If the current stabilizer generator being iterated through is some $s_Z$ that corresponds to some row of $H_Z$ and has weight lesser than or equal to 5, entangle the qubits in the support of $s_Z$ in the same order as how it would be measured in $M$.
			\item If the current stabilizer generator being iterated through is some $s_Z$ that corresponds to some row of $H_Z$ and has weight greater than 5, perform the following:
			\begin{enumerate}
				\item Let $\left\{s_{Z,j}\right\}_j$ be the set of new $Z$ stabilizer generators that correspond to the 1-cells in some $\overline{\calB}_i$ that corresponds to $s_Z$.
				\item Order the $s_{Z,j}$s in the same order as how the qubits in the support of $s_Z$ would have been measured in $M$.\footnote{Because each $s_{Z,j}$ has a direct correspondence to a qubit that $s_Z$ acts on, there is an obvious well-defined mapping and ordering.}
				\item Measure each of the $s_{Z,j}$ in the order that they were ordered in the previous step such that the qubits in the support of each $s_{Z,j}$ are entangled in any order.
			\end{enumerate}
		\end{enumerate}
		\item Iterate the new $X$ stabilizer generators that corresponds to the -1-cells in the $\overline{\calB}_i$s in any order and entangle the qubits in each of these $X$ stabilizer generators in any order.
	\end{enumerate}

	We now show that the effective distances of the cone code are preserved under the measurement schedule $M_C$.
	We start by showing that the effective $X$ distance is preserved, i.e. $\overline{d'}^{M_C}_X = \overline{d}^M_X$.
 	Let $\vec{x}$ be a non-trivial logical $X$ operator of the cone code.
	In addition, let $\vec{x}|_{\calA}$ be the restriction of $\vec{x}$ to the qubits in $\calA$ and $\vec{x}|_{\calB}$ be the restriction of $\vec{x}$ to the qubits in the 0-cells of $\calB = \bigoplus_i \overline{\calB}_i$.
	
	We argue that it is possible to clean $\vec{x}$ such that $\vec{x}|_{\calB} = \vec{0}$.
	To see this, consider the following process: for every qubit in the support of $\vec{x}$ that corresponds to a 0-cell of some $\overline{\calB}_i$, we can apply the $X$ stabilizer generator that is associated with that 0-cell via the chain map $f$ such that we obtain a new $\vec{x}'$ that has the same support as $\vec{x}$ but $\vec{x'}|_{\calB} = \vec{0}$.
	Since $\vec{x}'$ only has support on the qubits in $\calA$ and commutes with all the $Z$ stabilizer generators, it is a non-trivial logical operator of $\calQ$ and has weight at least $d_X$.
	In other words, any non-trivial logical operator of the cone code has weight at least $d_X$ since its existing support on the qubits in $\calA$ has to be a non-trivial logical operator of $\calQ$.
	By our assumption in the theorem statement, we require at least $\overline{d}^M_X$ elementary faults to generate $\vec{x}'$ if we measure the different $s_X$s according to $M$.
	In addition, the new $X$ stabilizer generators that correspond to the -1-cells of the $\overline{\calB}_i$s only have support on the qubits in the 0-cells of the $\overline{\calB}_i$s so the hook errors from their single-ancilla stabilizer measurement do not affect the qubits in $\calA$.
	Thus, the measurement order of the new $X$ stabilizer generators does not affect the effective $X$ distance of $\Cone\left(\calQ, f\right)$.
	Therefore, our construction of $M_C$ perfectly preserves the effective $X$ distance of the cone code.

	Next, we prove the effective $Z$ distance of the cone code.
	We claim that the effective $Z$ distance of the cone code	is given by $\overline{d'}^{M_C}_Z \geq \Omega\left(1/q_X\right)\overline{d}^M_Z \lambda_C$ where $\lambda_C$ is the soundness factor defined as the following for the chain complex obtained for $\Cone\left(\calQ, f\right)$:
	\[\lambda_C = \min\left(1, \min_i\left(\min_{u \in \left(\overline{\calB}_i\right)_0, \partial u = 0, u \neq 0}\left(\max_{v \in \left(\overline{\calB}_i\right)_1, u = \partial v}\frac{|u|}{|v|}\right)\right)\right).\]
	Note that $|v|$ can be also be interpreted as the number of qubits in $\calA$ that are in the support of the $Z$ stabilizer generators in $\calS_{Z,\calA^\ast}^c$ that are associated to the 1-chain $v$ by the chain map $f$. 
	Let $\vec{z}$ be a non-trivial logical $Z$ operator of the cone code.
	We argue that it is possible to clean $\vec{z}$ such that $\vec{z}|_{\calB} = \vec{0}$.
	To see this, consider the following process: for the qubits in the support of $\vec{z}|_{\calB}$, we can apply a set of $Z$ stabilizer generators that is associated to some 1-chain in $\calB$ such that we obtain a new $\vec{z}'$ where $\vec{z}'|_{\overline{\calB}} = \vec{0}$.
	This is because the zeroth homology group in $\calB$ is trivial which implies that every 0-cycle is a boundary of some 1-chain, i.e., any $Z$ pauli operator on the qubits in $\calB$ that commutes with the $X$ stabilizer generators is generated by the $Z$ stabilizer generators in $\calB$.
	Using the soundness factor defined as well as the fact that each 1-cell in the 1-chain in $\calB$ can be mapped to a qubit in $\calA$ by the chain map $f$, the ratio of the number of qubits in $\calB$ cleaned to the number of new qubits introduced in $\calA$ is at least $\lambda_C$.
	Therefore, the $Z$ distance of $\Cone\left(\calQ, f\right)$ is at least $\lambda_C d_Z$ since every $Z$ logical operator of $\calQ$ can be shortened by a factor of at most $\lambda_C$ by ``re-routing'' the 1-chain through the ``cones''.

	Using the $\lambda_C$ factor, we now attempt to bound a similar ratio but for the elementary faults instead of explicit qubit errors.
	Let $f'$ be an arbitrary minimal length fault path that generates $\vec{z}'$.
	By assumption, $f'$ contains at least $\overline{d}^M_Z$ elementary faults.
	Let $S'$ be the subset of $f'$ that contains elementary faults that result in $Z$ errors on the qubits in the support of $\calS_{Z,\calA^\ast}^c$.
	In other words, $S'$ is a set of elementary faults that are either data qubit errors or gate error / ancilla qubit errors that propagate to data qubits that lie in the boundary of the $Z$ stabilizer generators that have been removed from $\calA$.
	Similarly, we can define $f$ to be the fault path that generates $\vec{z}$.
	We define $S$ to be the subset of $f$ that contains elementary faults that result in the 0-cycles in $\calB$ in $\vec{z}$.
	Thus, the ratio that we hope to bound is $|S|/|S'|$.
	Then, we claim the following:
	\[\frac{|S|}{|S'|} \geq \Omega(1/q_X)\lambda_C.\]
	To see this, note that the number of elementary faults to generate a Pauli $Z$ operator on some 0-cycle $u$ in $\overline{\calB}_i$ is at least $\Omega(1/q_X) |u|$.
	This is easy to see from Lemma~\ref{lem:elementary_fault_propagation} and the fact that the $Z$ stabilizer generators in $\calS_{Z, \calB}$ have weight at most $q_X + O(1)$.
	In addition, the number of elementary faults to generate $v'$, i.e., a Pauli $Z$ operator on the qubits in the support of the $Z$ stabilizer generators in $\calS_{Z, \calA^\ast}^c$, is at most $|v'|$.
	Therefore, our claim naturally follows from the two observations above.
	Thus, we can ``re-route'' the fault path $f'$ that generates $\vec{z}'$ such that the ratio of the number of elementary faults in $S$ to the number of elementary faults in $S'$ is at least $\Omega(1/q_X)\lambda_C$.
	This gives us a fault path $f$ that generates $\vec{z}$ and is at most a factor of $\Omega(1/q_X)\lambda_C$ shorter than $f'$.
	Therefore, the effective $Z$ distance of the cone code is at least $\Omega(1/q_X)\lambda_C d_Z^M$.

	\textbf{Part 2: Adaptation of $M_C$ to $M_{\text{thick}}$ for the Thickened and Height-Chosen Coned Code.}
	By using Theorem~\ref{thm:effective_distance_preservation_thickened_code}, we can construct a single-ancilla stabilizer measurement schedule $M_{\text{thick}}$ for the thickened and height-chosen cone code $\Cone\left(\calQ, f, \ell\right)$ such that $d^{M_{\text{thick}}}_X = d^{M_C}_X$ and $d^{M_{\text{thick}}}_Z = d^{M_C}_Z \ell \geq \Omega\left(1/q_X\right) d^M_Z \ell \lambda_C$.

	\textbf{Part 3: Adaptation of $M_{\text{thick}}$ to $M'$ for the Reduced Cone Code.}
	Having cellulated the different $\overline{\calB}_i$s, we define a new soundness factor $\lambda$ in the same way as $\lambda_C$ but with respect to the new chain complexes after cellulation.
	To adapt $M_{\text{thick}}$ to $M'$, we can use the same schedule construction as in Part 1 for the $X$ stabilizer generators in $\calA$ and all the $Z$ stabilizer generators.
	For the $X$ stabilizer generators in $\overline{\calB}$, we can measure them in any order and entangle the qubits in their support in any order.
	Using the proof argument from Part 1 with the new soundness factor $\lambda$, we can show that the effective distances of $\calQ' = \Cone_R\left(\calQ, f, \ell\right)$ under the measurement schedule $M'$ are $d^{M'}_X = d^{M_{\text{thick}}}_X = d^M_X$ and $d^{M'}_Z \geq \Omega\left(1/q_X\right) d^M_Z \ell \lambda$.
\end{proof}

\section{Discussion and Outlook}\label{sec:discussion}
In this paper, we made progress towards understanding the fault tolerance of Hastings's weight reduction techniques and proved that the effective distance of the weight-reduced code is largely preserved by adapting known single-ancilla syndrome extraction circuits.
We also showed that the distance balancing technique in Ref.~\cite{evra2022decodable} preserves effective distance and generalize the result to higher-dimensional hypergraph product (HGP) codes.
Crucially, our results show that higher-dimensional HGP codes are great candidates for fault-tolerant quantum computation because of their robustness against hook errors in addition to their amenability to single-shot decoding and parallel logical gates.
In Hastings's work, the weight reduction techniques--\emph{copying}, \emph{gauging}, \emph{thickening and choosing heights}, and \emph{coning}--are applied in this specific order to reduce the weights of an arbitrary CSS quantum code to at most 5.
A naive application of Lemma~\ref{lem:elementary_fault_propagation} would allow us to conclude that we can preserve the effective distances up to at most a factor of 1/2 without having to be given some single-ancilla syndrome extraction circuit that preserves the effective distances of the original code up to some extent.
However, it is not always the case where we want to reduce the weights of our codes to the minimum weight possible and there might also be independent interest to utilize individual weight reduction techniques.
Our results thus further our understanding on how specific and independent applications of the techniques influence the effective distances of the quantum code.
While we do not expect Hastings's weight reduction techniques to produce QECCs with low weights that necessarily have competitive code parameters when compared to bivariate bicycle codes and lifted product codes, our work makes progress in advancing our understanding on the tradeoffs between quantum weights and effective distances.

While we have constructed several single-ancilla stabilizer measurement circuits for the different weight reduction techniques, it would be interesting to consider the optimality of the adapted syndrome measurement schedule and whether there exists schedules that are more efficient for specific quantum codes that are weight reduced.
Because weight reduction is crucial for practical implementation of quantum error correction, it would also be interesting to investigate how we might be able to preserve other interesting code properties beyond effective distance like transversal logical gates after weight reduction.
With regards to possible future work for weight reduction, it would be interesting to consider weight reduction for qudit subsystem CSS codes.
It may also be interesting for the community to explore a more general homological framework for weight reduction.

\section*{Acknowledgments}
\label{sec:acknowledgments}
We are grateful for the helpful discussions and feedback from Eric Sabo, Benjamin Ide, Anirudh Krishna, Qian Xu, Argyris Giannisis Manes, Xiaozhen Fu, Guo (Jerry) Zheng, Ruixi Huang, Victor Albert, and Daniel Gottesman.
We especially thank Christopher Pattison for initial discussions and introducing quantum weight reduction to us.
We also thank the anonymous QIP reviewers for encouraging us to address the effective distance preservation for coning and pointing out a mistake in our proof for copying and gauging.
We thank Matthew B. Hastings for allowing us to use the coning figure from Ref.~\cite{hastings2021quantum} in Fig.~\ref{fig:coning}.
SJST acknowledges funding and support from Joint Center for Quantum Information and Computer Science (QuICS) Lanczos Graduate Fellowship, MathQuantum Graduate Fellowship, and the National University of Singapore (NUS) Development Grant.
LS acknowledges funding and support from the NSF Graduate Research Fellowship Program.

\bibliographystyle{alphaurl}
\bibliography{bib/ref}

\end{document}